\documentclass[journal]{IEEEtran}

\makeatletter
\def\endthebibliography{%
	\def\@noitemerr{\@latex@warning{Empty `thebibliography' environment}}%
	\endlist
}
\makeatother

\usepackage[usenames,dvipsnames]{xcolor}
\usepackage{amsfonts}
\usepackage{amsmath,amsthm,amssymb,dsfont}
\usepackage{enumerate}
\usepackage[english]{babel}
\usepackage{graphicx}	
\usepackage[caption=false]{subfig}
\usepackage{url} 
\usepackage{todonotes}
\usepackage{bbm}
\usepackage{booktabs}
\usepackage{hyperref} 

\usepackage{amsmath,amssymb,amsthm,dsfont,mathrsfs}
\usepackage{makecell}
\usepackage{array,multirow,graphicx}
\usepackage{amsmath}
\usetikzlibrary{arrows.meta}

\newcommand{\rEm}{\mathrm{E}_{\max}}
\newcommand{\Var}{\mathrm{Var}}
\newcommand{\Cov}{\mathrm{Cov}}

\usepackage{algorithm}
\usepackage[noend]{algorithmic}
\usepackage{enumitem}
\usepackage{tikz}
\usetikzlibrary{chains}
\usetikzlibrary{fit}
\usepackage{pgflibraryarrows}		
\usepackage{pgflibrarysnakes}		

\usepackage{makecell}

\usepackage{epsfig}
\usetikzlibrary{shapes.symbols,patterns} 
\usepackage{pgfplots}
\pgfplotsset{compat=newest}

\usepackage{mathrsfs}

\usepackage{hyperref}
\hypersetup{colorlinks=true,citecolor=blue,linkcolor=blue,filecolor=blue,urlcolor=blue,breaklinks=true}

\usepackage{nicefrac}
\usepackage{mathtools}
\usepackage{mathtools}
\usepackage[utf8]{inputenc}
\usepackage[english]{babel}
\usepackage{amssymb,amsmath,amsthm}

\newtheorem{theorem}{Theorem}
\newtheorem{lemma}{Lemma}

\theoremstyle{definition}
\newtheorem{remark}{Remark}
\theoremstyle{definition}
\newtheorem{definition}{Definition}

\newcommand{\Ac}{\mathcal{A}}
\newcommand{\Bc}{\mathcal{B}}

\newcommand{\Dc}{\mathcal{D}}
\newcommand{\Ec}{\mathcal{E}}

\newcommand{\Gc}{\mathcal{G}}

\newcommand{\Sc}{\mathcal{S}}
\newcommand{\Tc}{\mathcal{T}}

\newcommand{\EE}{\mathbb{E}}
\newcommand{\PP}{\mathbb{P}}

\newcommand{\dmax}{\mathrm{d_{max}}}
\newcommand{\ER}{\mathrm{ER}}
\newcommand{\rE}{\mathrm{E}}
\newcommand{\PV}{\nu}
\newcommand{\kb}{\overline{k}}
\DeclarePairedDelimiter\ceil{\lceil}{\rceil}
\newcommand{\ID}{\mathrm{IND}}

\newcommand{\bigO}{\mathcal{O}}

\newcommand{\barD}{\overline{\mathcal{D}}}

\begin{document}

\title{A Fast Binary Splitting Approach for Non-Adaptive Learning of Erd\H{o}s--R\'enyi Graphs}
%
%

%

\author{Hoang Ta and Jonathan Scarlett

\thanks{This work is supported by the National University of Singapore under the Presidential Young Professorship scheme. \scriptsize (Corresponding author: Hoang Ta).}

\thanks{ Hoang Ta was with the Department of Computer Science, School of Computing, National University of Singapore (NUS), Singapore 117417. He is now with the Department of Computer Science, Hanoi University of Science and Technology, Vietnam (e-mail: hoang.taduy@hust.edu.vn).}
 

\thanks{Jonathan Scarlett is with the Department of Computer Science, Department of Mathematics, and Institute of Data Science, NUS, Singapore 117417 (e-mail: scarlett@comp.nus.edu.sg).}
}


\maketitle

\begin{abstract}
We study the problem of learning an unknown graph via group queries on vertex subsets, where each query reports whether at least one edge is present among the queried vertices. In general, learning arbitrary graphs with \(n\) vertices and \(k\) edges is hard in the non\mbox{-}adaptive setting, requiring \(\Omega \big(\min\{k^2\log n,\,n^2\}\big)\) tests even when a small error probability is allowed. We focus on learning Erd\H{o}s–R\'enyi (ER) graphs \(G\sim\ER(n,q)\) in the non\mbox{-}adaptive setting, where the expected number of edges is \(\kb=q\binom{n}{2}\), and we aim to design an efficient testing–decoding scheme, namely, a non-adaptive test design together with a decoding algorithm, achieving asymptotically vanishing error probability. Prior work (Li–Fresacher–Scarlett, NeurIPS 2019) presents a testing–decoding scheme that attains an order\mbox{-}optimal number of tests \(\bigO(\kb\log n)\) but incurs \(\Omega(n^2)\) decoding time, whereas their proposed sublinear-time algorithm incurs an extra $\log \kb \cdot \log n$ factor in the number of tests. We extend the binary splitting approach—recently developed for non\mbox{-}adaptive group testing—to the ER graph learning setting, and prove that the edge set can be recovered with high probability using \(\bigO(\kb\log n)\) tests while attaining decoding time \(\bigO(\kb^{1+\delta}\log n)\) for any fixed \(\delta>0\).
\end{abstract}


\IEEEpeerreviewmaketitle

\section{Introduction}
\label{sec:intro}

One of the central problems in learning theory and combinatorial inference is to recover the structure of an unknown graph from indirect observations. In many settings of interest, edges are not observed directly; instead, one has access to a query mechanism. In this paper, we focus on the case that we can query subsets of vertices and learn whether at least one edge is present among them. We refer to this task as \emph{graph learning via edge-detecting queries}. This problem arises in applications such as identifying which chemicals react with each other using tests that only detect whether any reaction occurs~\cite{bouvel2005combinatorial}, and it is closely related to the widely studied group testing problem \cite{aldridge2019group}.

The objective is to design a testing–decoding scheme that minimizes the number of queries while ensuring efficient recovery. Here, a testing–decoding scheme consists of both the test design (i.e., the collection of query subsets) and the decoder that reconstructs the graph from the observed test outcomes.
Two testing models are commonly considered. In the \emph{adaptive} setting, queries are issued in rounds and may be chosen based on previously observed outcomes; at each round, the next query is selected as a function of the outcomes so far until sufficient information has been accumulated for reconstruction. In this setting, the problem is well understood~\cite{johann2002group}, and can be solved using $\bigO(k \log n)$ queries for arbitrary graphs with \(n\) vertices and \(k\) edges. In the \emph{non\mbox{-}adaptive} setting, all query subsets are fixed in advance and evaluated in a single batch, with no dependence on intermediate outcomes. An impossibility result~\cite{abasi2019learning} shows that, in the worst case over graphs with a bounded number of edges, significantly more non\mbox{-}adaptive tests are required:  Any non\mbox{-}adaptive scheme that identifies arbitrary graphs with \(n\) vertices and \(k\) edges must use at least \(\Omega \big(\min\{k^2\log n,\,n^2\}\big)\) tests, even under a small-error criterion.

To circumvent the difficulty of worst-case graphs, a natural starting point is to study the Erd\"os–R\'enyi random graph model \(\ER(n,q)\), in which each possible edge appears independently with probability \(q\). Under this model, prior work~\cite{li2019learning} in the non\mbox{-}adaptive setting presents several designs that achieve asymptotically vanishing error with \(\bigO(\kb\log n)\) tests (where \(\kb = q \binom{n}{2}\) is the expected number of edges), but with decoding time that is not sublinear in \(n\) (e.g., \(\bigO(n^{2}\,\kb\log n)\)). Moreover, a general converse shows that any non\mbox{-}adaptive test design achieving asymptotically vanishing error must use at least \(\Omega \big(\kb\log (n^{2}/\kb)\big)\) tests, implying that the \(\bigO(\kb\log n)\) scaling is information-theoretically optimal in most regimes of \(\kb\). The same work also introduces a testing–decoding scheme based on the GROTESQUE algorithm from group testing~\cite{cai2017efficient}, achieving sublinear decoding time \(\bigO \big(\kb\log^{2}\kb+\kb\log n\big)\), but this comes at the expense of requiring \(\bigO \big(\kb\log \kb \cdot \log^{2} n\big)\) tests.

In this paper, we continue this line of work and propose a testing--decoding scheme that requires $\bigO(\kb \log n)$ tests while achieving a decoding time of $\bigO(\kb^{1+\delta}\log n)$ for any fixed $\delta>0$. This attains (nearly) the best of both objectives in terms of the number of tests and the decoding time. To achieve this goal, we extend the \emph{binary splitting} approach for non\mbox{-}adaptive group testing, introduced in~\cite{price2020fast}, to the graph-learning setting.  Despite the conceptual similarity, the details become significantly different and challenging in this setting for several inter-related reasons, e.g.:
\begin{itemize}
    \item A vertex being in a negative test is not in itself conclusive of anything. Rather, we can only detect \emph{non-edge pairs} in negative tests.
Thus, we need a careful design that ensures most non-edge pairs will be in at least one test \emph{without} any edge pairs, and this consideration of pairs complicates matters compared to that of individual items in standard group testing.
    \item When we group items into ``nodes'' at various levels following~\cite{price2020fast}, internal edges within such a node makes all of its tests positive and thus lacking in information on non-edge pairs.
    \item Our initial solution incurs $\kb^{1.5}$ dependence in the decoding time due to checking $\Omega(\sqrt{\kb})$ tests for each of $\Omega(\kb)$ pairs, and bringing the decoding time down to $\kb^{1+\delta}$ requires further refinements based on carefully splitting the graph learning problem into smaller subproblems.
\end{itemize}

\subsection{Related Work}
\label{sec:related_work}

The problem of learning a graph via edge-detecting queries can be viewed as a form of group testing with structural constraints~\cite[Section~5.8]{aldridge2019group}. In standard group testing, the goal is to identify \(K\) defective items among \(N\) items by querying subsets, and many variants have been studied (e.g., noisy vs.\ noiseless, as well as various pooling constraints). A particularly important distinction is that between the small-error setting, which allows an asymptotically vanishing failure probability (see~\cite{aldridge2019group} for a survey), and the zero-error setting, which requires exact identification in all instances (see~\cite{du1999combinatorial} for a survey). These two settings lead to very different test requirements: in the small-error case \(\bigO(K\log N)\) tests suffice, whereas in the zero-error case \(\Omega(\min\{N,\,K^2\})\) tests are  required. Graph learning with edge-detecting queries can be seen as a form of constrained group testing in which the ``items'' are potential edges and the admissible query sets are restricted by the graph structure. Throughout this work, we focus on the small-error setting and study the learning of Erd\H{o}s–R\'enyi graphs, which is a fundamental subclass of general graphs.

Early work on graph learning with edge-detecting queries first focused on finding a single hidden edge~\cite{aigner1986search,aigner1988searching}, and later on finding multiple edges~\cite{johann2002group} in a more general scenario in which the “defective’’ graph \(G\) is known to be a subgraph of a larger graph \(H\). In particular, an adaptive procedure was proposed that recovers \(G\) using \(k \log_2 \big(\tfrac{|E(H)|}{k}\big) + \bigO(k)\) tests, which is optimal up to the lower-order \(\bigO(k)\) term, where \(k\) is the number of edges of \(G\). Beyond this general setup, several works analyze specific families such as matchings, stars, and cliques~\cite{grebinski1998reconstructing,alon2004learning,alon2005learning}. While the adaptive setting is relatively well understood, the non\mbox{-}adaptive setting~\cite{abasi2019learning,kameli2018non} and adaptive strategies with a small number of stages~\cite{du2006pooling,bshouty2015linear,abasi2019learning} remain more challenging. In the Erd\H{o}s–R\'enyi model \(G\sim\ER(n,q)\), where the expected number of edges is \(\kb=q\binom{n}{2}\), prior work~\cite{li2019learning} in the non\mbox{-}adaptive setting shows that asymptotically vanishing error can be achieved with \(\bigO(\kb\log n)\) tests, but with decoding time that is at least quadratic in \(n\).  As noted above, they also propose a sublinear-time algorithm based on GROTESQUE \cite{cai2017efficient}, but it incurs an extra $\log \kb \cdot \log n$ term in the number of tests.  The problem has also been generalized to hypergraphs; see~\cite{angluin2006learning,angluin2008learning,d2016multistage,abasi2018non}, and an Erd\H{o}s–R\'enyi–type model for hypergraphs has likewise been explored~\cite{austhof2025non}.

The binary splitting approach in group testing was first developed in the adaptive setting~\cite{hwang1972method}, where it achieves near-optimal test complexity by recursively bisecting positive pools to isolate defectives. More recently, this method has been adapted to the non\mbox{-}adaptive setting—retaining the spirit of binary splitting through carefully designed test matrices and efficient decoders—and notably achieves $\bigO(k \log n)$ scaling in both the number of tests and decoding time~\cite{price2020fast,cheraghchi2020combinatorial, price2023fast,wang2023quickly}.

\subsection{Overview of Techniques and Contributions}

We extend the binary splitting approach to the non-adaptive setting of graph learning via edge-detecting queries in the Erd\H{o}s--R\'enyi model
\(\mathrm{ER}(n,q)\), where \(n\) is the number of vertices and \(q=\Theta\bigl(n^{-2(1-\theta)}\bigr)\) for some fixed
\(\theta\in(0,1)\). Let \(\kb\) denote the expected number of edges in \(G\sim \mathrm{ER}(n,q)\), namely \(\kb=q\binom{n}{2}\). By analogy with
group testing, we also refer to edges as defective pairs of vertices. We propose a testing--decoding scheme that uses \(\bigO(\kb\log n)\)
tests and, with high probability, recovers the edge set of \(G\) with decoding time \(\bigO(\kb^{1+\delta}\log n)\) for any fixed
\(\delta>0\). Our method is summarized as follows:

\paragraph{Binary splitting}
Building on the binary splitting idea in~\cite{price2020fast,cheraghchi2020combinatorial}, we organize the
vertices into a binary hierarchy of groups over levels \(\ell=\lceil \log_2 \sqrt{\kb}\,\rceil,\ldots,\log_2 n\), where
level~\(\ell\) has \(g=2^\ell\) blocks of size \(n/2^{\ell}\) vertices, and each level's groups are formed by splitting the previous level's blocks into two equal-size sub-blocks (see Figure \ref{fig:tree_structures}). 
At each level, we perform \(\Theta(\sqrt{\kb})\) independent
repetitions, and in each repetition, each block is assigned uniformly at random to one of \(\Theta(\sqrt{\kb})\) tests. This \(\sqrt{\kb}\)-by-\(\sqrt{\kb}\) design is based on balancing two desired properties:
\begin{itemize}
    \item[(i)] We would like most non-defective pairs to be placed together in at least one test (to give a chance for their non-defectivity to be detected).  Within each repetition this occurs with probability \(\Theta(1/\sqrt{\kb})\), and by having \(\Theta(\sqrt{\kb})\) repetitions, we can raise this to a suitably-designed constant.
    \item[(ii)] We would like the probability of a positive test to be bounded by a suitable constant strictly less than 1, so that the defective pairs do not ``drown out'' the non-defective pairs and make them undetectable.  Since there are $\bigO(\kb)$ defective pairs with high probability, we can readily verify that the \(\sqrt{\kb}\)-by-\(\sqrt{\kb}\) design indeed ensures this.
\end{itemize}
Moreover, the number of tests is controlled to $\bigO(\kb)$ per level, and thus \(\bigO(\kb\log n)\) overall.

The decoder works by successively iterating from coarse levels to fine levels. It maintains a set of possibly defective pairs, eliminates those ruled
out by negative tests, and refines the surviving pairs into child pairs at the next level. At each refinement step, every surviving pair
produces only a constant number of children, while each non-defective pair has a constant-order probability of being eliminated.  By controlling these constants, we can ensure that the `false positive' 
candidates evolve (at least heuristically) in a similar manner as a subcritical branching process: Although refinement creates new candidates, elimination dominates on
average, so the total number of possibly defective pairs remains under control across all levels.   Moreover, additional testing is done at the final level to identify the \emph{exact} set of defective pairs (as opposed to only a superset, which is what we obtain at the earlier levels).

We formalize the above intuition in Theorem~\ref{thm:main}, showing that the scheme recovers all edges with high probability using \(\bigO(\kb\log n)\) tests. The decoding time is \(\bigO(\kb^{1.5}\log n)\) when \(\theta>1/2\), and \(\bigO(\kb^{1.5}\log^{2}\kb\,\log n)\) when \(\theta\le 1/2\). The \(\kb^{1.5}\) factor comes from checking \(\Theta(\sqrt{\kb})\) candidate tests per retained pair.

\paragraph{Partitioning and permutations}
To overcome the \(\kb^{1.5}\) bottleneck in the decoding time, we partition the vertex set into
\(m=\kb^{(1-\gamma)/2}\) balanced parts, where \(\gamma \in \bigl(0,\min\{1,\frac{1-\theta}{3\theta}\}\bigr)\) is a fixed constant,
and apply the binary splitting design separately to each subproblem induced by the union of a pair of parts. Each such subproblem has expected edge count only
\(\bigO(\kb^\gamma)\), so the extra multiplicative term in the per-subproblem decoding cost drops from \(\sqrt{\kb}\) to \(\kb^{\gamma/2}\),
thus yielding a much smaller overall decoding time upon combining the \(\bigO(\kb^{1-\gamma})\) subproblems.

A key difficulty is that the partition must be suitably ``favorable'' for the local decoder to succeed, or at least for our particular analysis to work.  Informally, at the base level of the recursion,
it would be ideal if the blocks were sparse enough such that \emph{no block contains an internal edge}, as this prevents the undesirable scenario that placing a block in a given test automatically makes it positive, allowing us to apply our analysis techniques with sharper bounds.  
However, a naive union bound over all \(\bigO(\kb^{1-\gamma})\) subproblems turns out to be too loose for achieving this favorable property. To address this, we incorporate a
\emph{permutation amplification} step: We sample \(c=c(\gamma)\) independent random permutations from a pairwise independent family and run the
binary splitting design under each permutation. Each permutation is favorable with probability at least \(1 - \bigO(\kb^{-\gamma})\), and over \(c\) independent trials the probability that \emph{none} is favorable for a given subproblem decays polynomially in~\(\kb\). This makes the union bound over all subproblems go through, 
and the total number of tests remains
\(\bigO(\kb\log n)\). As stated in Theorem~\ref{thm:permutation_method}, this yields recovery of the edge set with high probability using \(\bigO(\kb\log n)\) tests and decoding time \(\bigO(\kb^{1+\delta}\log n)\) for any fixed \(\delta>0\) (where the hidden constant also depends on $\delta$).
\par 
Our results, as well as existing results for the Erd\H{o}s--R\'enyi  model, are summarized in Table \ref{tab:results}.

\begin{table*}[t]
  \centering
  \begin{tabular}{|c|c|c|}
    \hline
    \textbf{Reference} & \textbf{Number of tests} & \textbf{Decoding time} \\ \hline
    COMP~\cite{li2019learning} & $\bigO(\kb \cdot \log n)$ & $\Omega(n^2)$ \\ \hline
    GROTESQUE~\cite{li2019learning} & $\bigO(\kb \cdot \log \kb \cdot \log^2 n)$ & $\bigO(\kb \cdot \log^2 \kb + \kb \cdot \log n)$ \\ \hline
    {\bf Our Theorem \ref{thm:main}} & $\bigO(\kb \cdot \log n)$ & $\bigO(\kb^{1.5} \log n)$\footnotemark \\ \hline
    {\bf Our Theorem \ref{thm:permutation_method}} & $\bigO(\kb \cdot \log n)$ & $\bigO(\kb^{1+\delta} \cdot \log n)$ for any fixed $\delta>0$ \\ \hline
  \end{tabular}
   \vspace{0.5em}
  \caption{Overview of existing non-adaptive schemes for learning Erd\H{o}s--R\'enyi graphs $\ER(n,q)$. Here, $n$ is the number of vertices, and $\kb = q\binom{n}{2}$ is the expected number of edges. In the final row, the bound is \(\bigO(\kb^{1+\delta} \cdot \log n)\) for any fixed $\delta > 0$, but the hidden constant may depend on \(\delta\).}
  \label{tab:results}
\end{table*}

\footnotetext{We also reduce this to $\bigO(\kb \log^2 \kb \cdot \log^2 n)$ with ${\rm poly}(n)$-time pre-processing, but we are mainly focused on the case that there is no such pre-processing. }
\section{Problem Setup and Preliminaries}
\label{sec:setup}
We consider the problem of learning the edge structure of an unknown undirected graph \(G=(V,E)\).
The vertex set is \(V= [n] := \{1,2,\dots,n\}\), while the edge set \(E\subseteq \binom{V}{2}\) is random.
Specifically, under the Erd\H{o}s--R\'enyi (ER) model~\cite{li2019learning}, each potential edge \((i,j)\in\binom{V}{2}\) is included independently with probability \(q=q(n)\), so that \(G\sim \ER(n,q)\).
Once drawn, the graph remains fixed but is unknown, and the learner has no direct access to \(E\). 
Information about \(G\) is obtained through \emph{edge-detecting queries}:
given a subset \(S\subseteq V\), a query returns one bit indicating whether there exists an edge of \(E\) fully contained in \(S\). The objective is to design a \emph{non-adaptive} collection of such queries—fixed in advance—together with a decoder that reliably reconstructs \(E\). 

The problem setting is described in more detail as follows. The edge set is to be recovered through a sequence of binary queries.
A query is represented by a vector \(X\in\{0,1\}^n\), where \(X_i=1\) indicates the inclusion of vertex \(i\).
The corresponding outcome is
\[
  Y \;=\; \bigvee_{(i,j)\in E} \bigl( X_i \land X_j \bigr),
\]
that is, the output equals one if at least one edge of \(E\) is entirely contained in the selected set of vertices, and zero otherwise.
In the \emph{non-adaptive} setting, the collection of test vectors
\(X^{(1)},\dots,X^{(t)}\) must be fixed in advance, without access to intermediate outcomes, which renders the reconstruction task more difficult than in the adaptive case.

Given the observed outcomes \(\{Y^{(i)}\}_{i=1}^t\),
a decoder produces an estimate \(\widehat{G}=(V,\widehat{E})\).
The performance is measured by the error probability
\[
  P_e \;\coloneqq\; \PP \big[\,\widehat{E}\neq E\,\big],
\]
where the probability is taken over both the Erd\H{o}s--R\'enyi graph and the randomness of the test design.
The objective is to design a procedure such that \(P_e \to 0\) as \(n \to \infty\).
\medskip
\subsection{Sparsity Regime}
\label{sec:sparsity}
We focus on sparse graphs, which commonly arise in applications such as biological interaction networks, road networks, and sensor graphs.
Sparsity is parameterized by a constant \( \theta \in (0,1) \) via
\[
   q \;=\; \Theta \bigl(n^{-2(1-\theta)}\bigr),
\]
so that the expected number of edges satisfies
\[
   \kb \;=\; q\binom{n}{2} \;=\; \Theta \bigl(n^{2\theta}\bigr).
\]
Equivalently, as \( \theta \) ranges over \( (0,1) \), we have \(    n^{-2} \ll\; q \ll\; 1\), and thus $1 \ll \kb \ll n^2$, 
where \( f(n)\ll g(n) \) abbreviates \( f(n)=o\bigl(g(n)\bigr) \).

These sparsity regimes are well studied in random-graph theory and information-theoretic learning. 
For example, the very sparse case \( \theta \leq \tfrac{1}{2} \) is relevant to connectivity thresholds and percolation~\cite{bollobas2011random}, and the denser range \( \theta>\tfrac{1}{2} \) appears in social-network modeling and large-scale biological systems~\cite{newman2003structure}.
\subsection{Mathematical and Computational Assumptions}
\label{sec:assumption}
Throughout the paper, we work in the unit-cost word-RAM model. With \(n\) vertices and \(T\) tests, reading any integer in \(\{1,\dots,n\}\), performing basic arithmetic on such integers, and retrieving any test outcome (indexed by \(\{1,\dots,T\}\)) each take \(\bigO(1)\) time.

Without loss of generality, we assume that \( n \) is a power of two. When this is not the case, the graph can be augmented by adding up to \( 2^{\lceil \log_2 n \rceil} - n \) dummy vertices so that the total number of vertices equals the next power of two. Since these additional vertices are isolated and introduce no edges, they do not affect the test outcomes or the recovery procedure. Consequently, all subsequent results remain valid under this assumption.

\subsection{Level Graphs}
\label{sec:level_graph}

As we have already hinted, our algorithm considers groups of nodes at various levels that are always tested together, starting with larger groups that are recursively split into smaller groups until the final level containing only singletons.  We will provide further details when describing our algorithm; see in particular Figure \ref{fig:tree_structures} in Section \ref{sec:binary_spliting}.  For now, we only give a brief description to the extent required in the next subsection.


For each \emph{level} \(\ell \in \{\ceil{\log_2 \sqrt{\kb}},\dots,\log_2 n\}\), let \(g \coloneqq 2^{\ell}\), and let \(\{\mathcal{G}_1,\dots,\mathcal{G}_g\}\) be a balanced partition of \(V\) into \(g\) blocks, each of size \(n/g\), where
\[
  \mathcal{G}_i \coloneqq \big\{\, (i-1)\tfrac{n}{g}+1,\ \dots,\ i\tfrac{n}{g} \,\big\}, \qquad i \in [g].
\]
We call a block $\Gc_i$ \emph{defective} if it contains at least one edge of $G$, and we call $(\Gc_i,\Gc_j)$ a \emph{defective pair} if $\Gc_i \cup \Gc_j$ contains at least one edge of $G$.  Then, we introduce the following quantities that we will use regularly:
\begin{itemize}
   \item $\PV_g$: the number of defective blocks among the $g$ blocks.
   \item $G_g=(V_g,E_g)$: the induced graph at level $\ell$, where $g=2^\ell$, $V_g=[g]$, and an edge $(i,j)$ is present if and only if $(\Gc_i,\Gc_j)$ is defective.  We refer to $G_g$ as the \emph{block graph} at level $\ell$. 
   \item $d_g(i)$: the degree of a non-defective block $\Gc_i$ in the block graph $G_g$.
   \item $d(G_g)$: the maximum degree in $G_g$ among all non-defective blocks.
\end{itemize}

\subsection{Typical Graphs}
\label{sec:typical_graphs}

In order to analyze the performance of our decoding procedures, it is useful to restrict attention to graphs that behave in a “regular” manner.  
Specifically, we identify a high-probability subset of Erd\H{o}s--R\'enyi graphs in which the number of edges, as well as various structural quantities across multiple partition levels, remain within controlled bounds.  
These graphs, which we refer to as the \emph{typical set of graphs}, form the basis of our subsequent analysis.

\begin{definition} \label{def:typical}
    Let \( (\epsilon_n)_{n \in \mathbb{N}} \) be a sequence with \( \epsilon_n \to 0 \). We define the \( \epsilon_n \)-typical set of graphs $\Tc(\epsilon_n)$ as the collection of graphs $G$ satisfying the following:  
    
    \begin{enumerate}[label=(\roman*)]
        \item {
        The number of edges is close to its expectation, i.e.,
    \begin{equation}
     (1-\epsilon_n)\kb \leq k \leq (1+\epsilon_n)\kb,  \text{ where }k=|E|\,. \label{eq:k_concentraton}
    \end{equation}
        }
        \item{
        At every level $\ell \in \{\ceil{\log_2 \sqrt{\kb}},\dots,\log_2 n\}$ with $g=2^{\ell}$, the quantities $|E_g|$, $\PV_g$, and $d(G_g)$ (see Section \ref{sec:level_graph}) are bounded as follows:
    \begin{alignat}{2}
      |E_g| &\le \rEm &&\coloneqq 
      \begin{cases}
        4\kb & \theta > \tfrac{1}{2},\\
        2\kb \log^2 \kb & \theta \le \tfrac{1}{2},
      \end{cases} \label{eq:number_edges_level} \\[1em]
      \PV_g & \le \PV_{\max} &&\coloneqq 
      \begin{cases}
        \dfrac{2\kb}{g} & \theta > \tfrac{1}{2},\\
        2\sqrt{\kb} & \theta \le \tfrac{1}{2},
      \end{cases} \label{eq:number_defective_level} \\[1em]
      d(G_g) &\le \mathrm{d}_{\max} &&\coloneqq 
      \begin{cases}
        \dfrac{10\kb}{g} & \theta > \tfrac{1}{2},\\
        8\sqrt{\kb} & \theta \le \tfrac{1}{2}.
      \end{cases} \label{eq:degree_level}
    \end{alignat}
        }
    \end{enumerate}
\end{definition}


Condition (i) ensures concentration of the global edge count, while condition (ii) guarantees regularity across different scales of partition: The number of induced edges, the number of defective blocks, and the maximum degree of non-defective nodes are all bounded in a certain manner depending on the sparsity regime.  
While condition~(i) was adopted in~\cite{li2019learning}, condition (ii) is much more specific to our algorithm.  The following lemma shows that both conditions hold with high probability.

\begin{lemma} \label{lem:typical}
Fix \( \theta \in (0,1) \), and let \( G \sim \mathrm{ER}(n, q) \) with \( q = \Theta(n^{-2(1-\theta)}) \).  
Then there exists a sequence \( \epsilon_n \to 0 \) such that
\[
  \PP[G \in \Tc(\epsilon_n)] \to 1 \quad \text{as } n \to \infty.
\]
\end{lemma}

The proof is ultimately mainly based on standard concentration inequalities and related tools, but the details are rather technical, so are deferred to Appendix \ref{sec:pf_typical}.


\section{Binary Splitting Approach for Graph Learning}
\label{sec:binary_spliting}

We now present our first algorithm for non-adaptive graph recovery. The method builds upon ideas from hierarchical binary splitting for standard group testing~\cite{price2020fast,cheraghchi2020combinatorial}. It consists of the construction of a sequence of tests, and a decoding procedure that recovers the edge set from the test outcomes. The main idea is to partition the vertex set hierarchically and perform tests on blocks  to determine whether edges exist between them.  Here, placing a block into a test means placing all of its items, essentially creating a ``super-item''.  

As illustrated in Figure \ref{fig:tree_structures}, at each level of the hierarchy, we divide the vertex set into smaller blocks (groups).  We then use randomized test assignments to probe the presence or absence of edges. Negative test outcomes allow us to rule out large numbers of candidate edge pairs efficiently.  When two larger blocks (say $\Gc_1,\Gc_2$ in generic notation) are not ruled out and thus believed to (potentially) have an edge between them, and we split them both into smaller blocks (say $\Gc_{1}^{(L)},\Gc_{1}^{(R)}$ and $\Gc_{2}^{(L)},\Gc_{2}^{(R)}$ of half the size), we treat all combinations of resulting edges as potentially possible: (i) an internal edge in one of the 4 smaller groups; (ii) a ``formerly internal'' edge between $\Gc_{i}^{(L)}$ and $\Gc_{i}^{(R)}$ for $i \in \{1,2\}$; or (iii) a ``cross-group'' edge between some $\Gc_{1}^{(\cdot)}$ and some $\Gc_{2}^{(\cdot)}$.

\begin{figure}[t]
  \centering
  \includegraphics[width=1.0\linewidth]{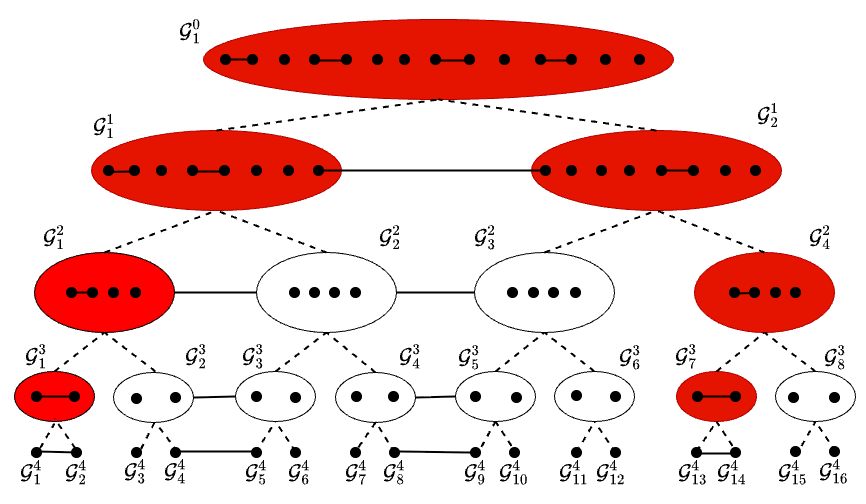} 
 \caption{Example tree structures at levels $\ell\in\{0,1,2,3,4\}$ for graph learning with $n=16$ and $k=4$. Defective blocks are colored red, and edges between blocks are indicated with a solid line.}
  \label{fig:tree_structures}
\end{figure}

\subsection{Testing Procedure}
\label{sec:testing}

As outlined above (and in Section \ref{sec:level_graph}), we construct a binary tree of vertex groups across levels \( \ell = \ceil{\log_2 \sqrt{\bar{k}}}, \ceil{\log_2 \sqrt{\bar{k}}} + 1, \dotsc, \log_2 n \). Each node in the tree corresponds to a block of vertices. At level \( \ell \), there are \( 2^{\ell} \) blocks \( \{\Gc^{(\ell)}_1, \dotsc, \Gc^{(\ell)}_{2^{\ell}}\} \), each containing \( n/2^{\ell} \) vertices.  We perform randomized testing as follows for some constants \( C_1, C_2 > 0 \) to be chosen later:
\begin{itemize}
    \item Create \( C_2 \sqrt{\kb} \) sequences of tests (``repetitions'' or ``iterations''), each of length \( C_1 \sqrt{\kb} \);
    \item Within each repetition, each group is assigned uniformly at random to one of these \( C_1 \sqrt{\kb} \) tests.
\end{itemize}
Intuitively, the reason for choosing $\bigO(\sqrt{\kb})$ tests per repetition and $\bigO(\sqrt{\kb})$ repetitions is that this amounts having $\bigO(\kb)$ tests total and an $\bigO(1/\kb)$ probability of a given test containing a given \emph{pair} of nodes.  These both match the scalings that arise in standard group testing \cite{price2020fast}, in particular ensuring that the probability of a given test being positive is bounded away from 0 and 1, implying non-vanishing entropy per test.  See Algorithm \ref{alg:testing} for a full description of the test design.

\begin{algorithm*}[t]
	\caption{Testing Procedure}
	\label{alg:testing}
	\begin{algorithmic}[1]
		\REQUIRE Number of vertices \( n \), average number of edges \( \kb \), constants $C_1>27,C_2 = C_1^2,C'>3$.
		\STATE Set \( \ell_{\min} \gets \ceil{\log_2 \sqrt{\bar{k}}} \)
		\FOR{each level \( \ell = \ell_{\min}, \dotsc, \log_2 n - 1 \)}
		\FOR{each iteration in \( \{1, \dotsc, C_2 \sqrt{\kb}\} \)} 
		\STATE Initialize a sequence of $C_1 \sqrt{\kb}$ tests
		\FOR{each block \( j = 1, \dotsc, 2^{\ell} \)}
		\STATE Assign block \( \Gc^{(\ell)}_j \) to a randomly chosen test among the $C_1 \sqrt{\kb}$ tests
		\ENDFOR
		\ENDFOR
		\ENDFOR
		\STATE At the final level $\ell = \log_2 n$, each block degenerates to a singleton, i.e., $\Gc_{j}^{(\ell)} = \{j\}$. Repeat steps 3--6 for $C' \log n$ rounds, with each round consisting of $C_1 C_2 \kb$ tests.
	\end{algorithmic}
\end{algorithm*}


\subsection{Decoding Procedure}
\label{sec:decoding}

The decoder is deterministic and proceeds in a coarse-to-fine manner over levels \(\ell= \ell_{\min}, \dots, \log_2 n \). At level \(\ell\), we maintain a set of \emph{possible defectives} \(\mathcal{PD}^{(\ell)}\), whose elements are block pairs \((\Gc_i,\Gc_j)\) that may still contain a true edge at the current resolution. A pair \((\Gc_i,\Gc_j)\) is \emph{discarded} from \(\mathcal{PD}^{(\ell)}\) if there exists a negative test that contains both \(\Gc_i\) and \(\Gc_j\). When no negative test contains both \(\Gc_i\) and \(\Gc_j\) we refine this pair: each block splits into two children, \(\Gc_i \to \{\Gc_i^{(L)},\Gc_i^{(R)}\}\) and \(\Gc_j \to \{\Gc_j^{(L)},\Gc_j^{(R)}\}\), and we add to \(\mathcal{PD}^{(\ell+1)}\) the six child–child pairs
\begin{equation}
\begin{gathered}
(\Gc_i^{(L)},\Gc_j^{(L)}),
(\Gc_i^{(L)},\Gc_j^{(R)}),
(\Gc_i^{(R)},\Gc_j^{(L)}),\\
(\Gc_i^{(R)},\Gc_j^{(R)}),
(\Gc_i^{(L)},\Gc_i^{(R)}),
(\Gc_j^{(L)},\Gc_j^{(R)}).
\end{gathered}
\label{eq:group_split}
\end{equation}

The four ``cross pairs'' capture edges potentially connecting the two blocks, while the two ``sibling pairs'' capture edges internal to either block.  Note that we do not need to explicitly keep track of which groups (potentially) have an internal edge, though those can still have a significant impact on our analysis due to all pairs involving that group being (potentially) defective. 
At the final level, any surviving vertex pair not eliminated by a negative test is declared an edge. The details are presented in Algorithm~\ref{alg:decoding}.

\begin{algorithm*}[t]
	\caption{Decoding Procedure}
	\label{alg:decoding}
	\begin{algorithmic}[1]
		\REQUIRE Outcomes of $t$ non-adaptive tests \( \{Y^{(i)}\}_{i=1}^t \) from Algorithm~\ref{alg:testing}, number of vertices \( n \), average number of edges \( \kb \)
		\STATE Initialize candidate set $\mathcal{PD}^{(\ell_{\min})} 
		= \{ (\Gc^{(\ell_{\min})}_i,\Gc^{(\ell_{\min})}_j) : i,j \in [2^{\ell_{\min}}] ,\, i \neq j \}$
		\FOR{each level \( \ell = \ell_{\min}, \dotsc, \log_2 n - 1 \)}
		\FOR{each pair \( (\Gc_i,\Gc_j) \in \mathcal{PD}^{(\ell)} \)}
		\IF{no negative test contains both \( \Gc_i \) and \( \Gc_j \)}
		\STATE Add all child pairs generated from $\Gc_i$ and $\Gc_j$ to \( \mathcal{PD}^{(\ell+1)} \) (see \eqref{eq:group_split})
		\ENDIF
		\ENDFOR
		\ENDFOR
		\STATE Let \( \widehat{E} \) be the set of pairs in $\mathcal{PD}^{(\log_2 n)}$ that are not included in any negative test at the final level
		\STATE Return graph estimate \( \widehat{G} = (V, \widehat{E}) \)
	\end{algorithmic}
\end{algorithm*}

%
%
\subsection{Algorithmic Guarantees}
\label{sec:algorithmic_guarantees}

In this section, we analyze the correctness and efficiency of our approach. We show that Algorithms~\ref{alg:testing} and~\ref{alg:decoding} use an order-optimal number of tests and, with high probability, recover the entire edge set of the underlying Erd\H{o}s--R\'enyi graph. We also establish the computational complexity of the decoding procedure.

    
\begin{theorem}
\label{thm:main}
Fix \( \theta \in (0,1) \), and let \( G \sim \ER(n, q) \) with \( q = \Theta \big(n^{-2(1 - \theta)}\big) \).
Let \( \kb = q \binom{n}{2} = \Theta \big(n^{2\theta}\big) \) denote the expected number of edges.
There exist constants \( C_1 > 27 \), \( C_2 = C_1^2 \), and \( C' > 3 \) such that the testing (Algorithm~\ref{alg:testing})–decoding (Algorithm~\ref{alg:decoding}) scheme achieves \( P_e \to 0 \) as \( n \to \infty \) with \( \bigO(\kb \log n) \) tests, and the decoding time is as follows, with probability \( 1 - o(1) \):
\begin{enumerate}[label=(\alph*)]
    \item If \( \theta > 1/2 \):
    \begin{itemize}
        \item The decoding time is \( \bigO \left(\kb^{1.5} \log n\right) \).
        \item With $\bigO(n^2 \sqrt{\kb} \log n)$ time pre-processing,\footnote{This refers to computation that can be done before observing any tests results, which is ``offline'' and thus may have more flexibility in taking longer.} the decoding time improves to \( \bigO \left(\kb \cdot \log^2 n\right) \).
    \end{itemize}

    \item If \( \theta \le 1/2 \):
    \begin{itemize}
        \item The decoding time is \( \bigO \left(\kb^{1.5} \log^2 \kb \log n\right) \).
        \item With $\bigO(n^2 \sqrt{\kb} \log n)$ time pre-processing, the decoding time improves to \( \bigO \left(\kb \log^2 \kb \cdot \log^2 n\right) \).
    \end{itemize}
\end{enumerate}
\end{theorem}
\begin{remark}
Recall that our analysis is carried out under the sparse Erd\H{o}s--R\'enyi regime $q=\Theta(n^{-2(1-\theta)})$, equivalently $\kb=\Theta(n^{2\theta})$, for a fixed \(\theta\in(0,1)\). In particular, the error term arising from the union bound over levels is of order \(\bigO(\frac{\log n}{\sqrt{\kb}})\), which is \(o(1)\) in this regime. We do not claim the same guarantee for very small values of \(\kb\) such as \( \kb = \bigO(\log n) \).
\end{remark}
From Lemma~\ref{lem:typical}, the Erdős–Rényi graph lies in the typical set of graphs (see Section~\ref{sec:typical_graphs}) with probability tending to $1$. 
Therefore, it suffices to prove Theorem~\ref{thm:main} conditioned on an arbitrary typical graph \( G \in \Tc(\epsilon_n) \), where \( \epsilon_n = o(1) \). 
All subsequent analysis proceeds under the implicit assumption that the graph belongs to the typical set.  Before proving Theorem~\ref{thm:main}, we introduce the following notation that will be used throughout its proof.




\begin{itemize}[leftmargin=*]
  \item { Fix a level \(\ell\) and set \(g \coloneqq 2^{\ell}\). For \(v \in [g]\), let \(\Gc_v\) denote the \(v\)-th node, i.e., the \(v\)-th group of vertices in \(G\) at level \(\ell\). A node \(\Gc_v\) is \emph{non-defective} if \(\Gc_v\) contains no (internal) edge. A pair \((\Gc_u,\Gc_v)\) with \(u \neq v\) is \emph{non-defective} if the induced subgraph on \(\Gc_u \cup \Gc_v\) contains no edges. For each node $u$, denote by $h(u) \in \{1,2,\dots,C_1\sqrt{\kb}\}$ the index of the test containing $u$ in a given iteration (out of $C_2\sqrt{\kb}$ iterations). For a single random test, let $Y$ denote the test outcome.}
  
  
  
  \item  {For a non-defective pair \((\Gc_u,\Gc_v)\) at level \(\ell\) with \(u \neq v\), let \(\Ec_{uv}\) be the event that \((\Gc_u,\Gc_v)\) is \emph{not} identified at level \(\ell\) in Algorithm~\ref{alg:decoding}, and let \(\rE_{uv}\) be its indicator random variable.
  The dependence of these quantities on $\ell$ is left implicit. We write $\EE_{\ell}[\cdot]$ for conditional expectation given all test placements at earlier levels.}
  
\end{itemize}

We first estimate the size of $\mathcal{PD}^{(\ell)}$ for all levels $\ell \in \{\ell_{\min},\dots,\log_2 n \}$.

\begin{lemma}
\label{lem:expectation}
Under the preceding setup and definitions, suppose that $|\mathcal{PD}^{(\ell)}| \leq 12 \rEm$, for $\rEm$ given in Eq~\eqref{eq:number_edges_level}. 
Then, for any $C_1 \geq 27$ and $C_2 = C_1^2$, we have
\[
\EE_{\ell} \left[\sum_{u,v} \rE_{uv} \right] \leq \frac{\rEm}{2},
\]
where the sum is over all non-defective pairs in $\mathcal{PD}^{(\ell)}$. 
\end{lemma}

\begin{proof}
It suffices to show that, for any non-defective pair $(\Gc_u,\Gc_v)$,
\begin{equation}
\label{eq:bound_pair_not_identify}
    \PP[\Ec_{uv}] \leq \frac{1}{C} \, ,
\end{equation}
for some constant $C \geq 24$. According to the testing procedure, a given non-defective pair $(\Gc_u,\Gc_v)$ fails to be identified correctly in a given sequence of $C_1\sqrt{\kb}$ tests if either:
\begin{itemize}
  \item $\Gc_u$ and $\Gc_v$ are not assigned to the same test, which occurs with probability $1 - \frac{1}{C_1\sqrt{\kb}}$;
  \item $\Gc_u$ and $\Gc_v$ are assigned to the same test, but the outcome is positive because of other edges.
\end{itemize}
Hence, the probability that $(\Gc_u,\Gc_v)$ is missed in all $C_2\sqrt{\kb}$ rounds, denoted by $\PP[\Ec_{uv}]$, is
\begin{equation}
\label{eq:prob_non_defective_pair}
    \left(1 - \frac{1}{C_1\sqrt{\kb}} + \frac{1}{C_1\sqrt{\kb}} \PP[Y=1 \mid h(u)=h(v)] \right)^{C_2\sqrt{\kb}}.
\end{equation}

Let $\Ac_1$ be the event that there exists some $w$ with $h(w)=h(u)=h(v)$ such that $(\Gc_u,\Gc_w)$ or $(\Gc_v,\Gc_w)$ is defective, and let $\Ac_2$ be the event that some other defective pair (containing neither $u$ nor $v$) is present in the test. Then
\[
\begin{aligned}
\PP[Y=1 \mid h(u)=h(v)]
&\le \PP[\Ac_1 \mid h(u)=h(v)] \\
&\quad + \PP[\Ac_2 \mid h(u)=h(v)] \, .
\end{aligned}
\]

From Lemma~\ref{lem:typical}, each non-defective block $u$ has at most $d_{\max}$ defective neighbors (see~\eqref{eq:degree_level}), so
\[
\begin{aligned}
  &\PP[\Ac_1 \mid  h(u)=h(v)] \\
  &\leq \frac{2d_{\max}}{C_1\sqrt{\kb}} 
  \leq \frac{2}{C_1\sqrt{\kb}} \max \left\{ \frac{10\kb}{g}, 8\sqrt{\kb} \right\} 
  \leq \frac{20}{C_1},
\end{aligned}
\]
since $g \in [\sqrt{\kb},n]$ due to the fact that $\ell \in [\log_2 \sqrt{\kb},\log_2 n]$.

Next, consider one of the \(C_1\sqrt{\kb}\) tests that contain both \(u\) and \(v\). Let $\Bc_1$ be the event that some defective block belongs to the test, and $\Bc_2$ be the event that some defective pair $(\Gc_{w},\Gc_{w'})$ is included in the test with both $\Gc_{w}$ and $\Gc_{w'}$ being non-defective i.e., edges appear only between $\Gc_{w}$ and $\Gc_{w'}$, but not internally. Since there are at most $k$ edges in $G$, there are at most $k \le (1+\epsilon_n)\kb$ such pairs. Moreover, from~\eqref{eq:number_defective_level} used in Lemma~\ref{lem:typical}, there are at most $\PV_{\max}$ defective blocks among $\{\Gc_1,\Gc_2,\dots,\Gc_g \}$. Thus,
\begin{align}
\PP[\Ac_2 \mid  h(u)=h(v)]  
   &\leq \mathbb{P}[\Bc_1] + \mathbb{P}[\Bc_2] \notag\\
   &\leq \PV_{\max} \cdot \frac{1}{C_1 \sqrt{\kb}} 
       + k \cdot \frac{1}{C_1^2 \kb} \notag\\
   &\leq \frac{1}{C_1 \sqrt{\kb}} 
       \cdot \max \left\{ \tfrac{2\kb}{g},\, 2\sqrt{\kb} \right\} 
       + \frac{1+\epsilon_n}{C_1^2} \notag\\
   &\leq \frac{2}{C_1} + \frac{1+\epsilon_n}{C_1^2} 
       \leq \frac{3}{C_1}\, , \label{eq:bound-Y}
\end{align}
where we again used the fact that $g \in [\sqrt{\kb}, n]$. Combining the bounds into~\eqref{eq:prob_non_defective_pair} and recalling the choice $C_2 = C_1^2$, we obtain
\[
\begin{aligned}
    \PP[\Ec_{uv}] &\leq \left(1 - \frac{1}{C_1\sqrt{\kb}} + \frac{23}{C_1^2\sqrt{\kb}} \right)^{C_2\sqrt{\kb}}\\
  &\leq \exp \left(-\frac{C_2(C_1-23)}{C_1^2}\right)\\
  &= \exp\big(-(C_1-23)\big).
\end{aligned}
\]
This establishes the claim, since this is at most $\exp(-4)< \frac{1}{24}$ when $C_1 \geq 27$.
\end{proof}

%
%
\begin{lemma}
\label{lem:variance}
There exist choices of $C_1$ and $C_2$ such that the following holds: Conditioned on the $\ell$-th level having $|\mathcal{PD}^{(\ell)}| \leq 12 \rEm$, we have 
\begin{align*}
   \Var_{\ell} \left[\sum_{u,v} \rE_{uv} \right] \leq \bigO \left( \frac{\rEm^2}{\sqrt{\kb}} \right),
\end{align*} 
where the sum is taken over all non-defective pairs in $\mathcal{PD}^{(\ell)}$.
\end{lemma}

\begin{proof}
	We first upper bound the covariance $\Cov[\rE_{uv}, \rE_{u'v'}]$ for any non-defective pairs $(\Gc_u, \Gc_v)$ and $(\Gc_{u'}, \Gc_{v'})$. We have
	\begin{equation}
            \label{eq:cova_formula}
	    \Cov[\rE_{uv}, \rE_{u'v'}] = \PP[\Ec_{uv} \cap \Ec_{u'v'}] - \PP[\Ec_{uv}] \PP[\Ec_{u'v'}] \,.
	\end{equation}
Consider a non-defective pair $(\Gc_u, \Gc_v)$ at level~$\ell$. Let $\Dc_{uv}$ be the event that the pair $(\Gc_u, \Gc_v)$ is not identified in the first $C_1 \sqrt{\kb}$ tests, and let $\barD_{uv}$ denote its complement. We consider the following cases:
	
	\begin{enumerate}
		\item \textbf{All nodes $u, v, u', v'$ are distinct.} We write
             \begin{equation}\label{eq:lem_cova_prob_Duv_cap}
               \PP[\Dc_{uv} \cup \Dc_{u'v'}] = 1 - \PP[\barD_{uv} \cap \barD_{u'v'}] \, .
             \end{equation}
  %
  
A pair $(u,v)$ is identified when $u$ and $v$ appear together in a test 
with a negative outcome. Therefore, the event $\barD_{uv} \cap \barD_{u'v'}$ 
can only occur when $u$ and $v$ have the same test placement, and similarly 
for $u'$ and $v'$. Letting
\begin{align*}
    \mathcal{C}_1 &= \bigl\{h(u) = h(v),\ h(u') = h(v'),\ h(u) \neq h(u')\bigr\}, \\
    \mathcal{C}_2 &= \bigl\{h(u) = h(v) = h(u') = h(v')\bigr\},
\end{align*}
this yields
\begin{align}
    &\PP\bigl[\barD_{uv} \cap \barD_{u'v'}\bigr] \\
    &= \frac{C_1 \sqrt{\kb}\,(C_1 \sqrt{\kb} - 1)}{(C_1\sqrt{\kb})^{4}} 
       \PP\bigl[\barD_{uv} \cap \barD_{u'v'} \mid \mathcal{C}_1\bigr] \notag \\
    &\quad + \frac{1}{(C_1 \sqrt{\kb})^3} 
       \PP\bigl[\barD_{uv} \cap \barD_{u'v'} \mid \mathcal{C}_2\bigr] \notag \\
    &= \frac{1}{(C_1 \sqrt{\kb})^2}
       \PP\bigl[\barD_{uv} \cap \barD_{u'v'} \mid \mathcal{C}_1\bigr] \notag \\
    &\quad + \bigO\!\left(\frac{1}{\kb^{3/2}}\right).
    \label{eq:lem_cova_Dc_uv_cap}
\end{align}
		
        Set $\gamma \coloneqq \PP\bigl[\barD_{uv} \mid h(u) = h(v)\bigr]$, $\beta \coloneqq \PP\bigl[\barD_{uv} \cap \barD_{u'v'} \mid \mathcal{C}_1\bigr],$
and observe that $\PP[\Dc_{uv} \mid h(u) = h(v)] = 1 - \gamma$.
        Since the event $h(u) \ne h(v)$ directly implies $\Dc_{uv}$, it follows that
\begin{align}
    \PP[\Dc_{uv}] 
    &= 1 - \frac{1}{C_1 \sqrt{\kb}} 
       + \frac{\PP[\Dc_{uv} \mid h(u) = h(v)] }{C_1 \sqrt{\kb}} \notag \\
    &= 1 - \frac{\gamma}{C_1 \sqrt{\kb}}\, ,  
    \label{eq:lem_cova_D_uv} \\[4pt]
    \PP[\Ec_{uv}] 
    &= \left(1 - \frac{\gamma}{C_1 \sqrt{\kb}} \right)^{C_2 \sqrt{\kb}} \,.
    \label{eq:lem_cova_E_uv}
\end{align}
        
		Moreover, from~\eqref{eq:lem_cova_prob_Duv_cap}, we have  
		\begin{equation} \label{eq:lem_cova_D_uv_cup}
		    \PP[\Dc_{uv} \cup \Dc_{u'v'}] = 1 - \frac{\beta}{C_1^2 \kb} + \bigO\left( \frac{1}{\kb^{3/2}} \right) \,.
		\end{equation}

        We now write $\PP[\Dc_{uv} \cap \Dc_{u'v'}] = \PP[\Dc_{uv}] + \PP[\Dc_{u'v'}] 
- \PP[\Dc_{uv} \cup \Dc_{u'v'}]$; then, from~\eqref{eq:lem_cova_D_uv} 
and~\eqref{eq:lem_cova_D_uv_cup}, it follows that
\begin{align}
    \PP[\Dc_{uv} \cap \Dc_{u'v'}] \notag 
    &=  1 - \frac{2\gamma}{C_1 \sqrt{\kb}}  
       + \frac{\beta}{C_1^2 \kb} + \bigO\left( \frac{1}{\kb^{3/2}} \right) \notag \\
    &= 1 - \frac{1}{C_1 \sqrt{\kb}} \left( 2\gamma - \frac{\beta}{C_1 \sqrt{\kb}} 
       + \bigO\left( \frac{1}{\kb} \right) \right) \,.
    \label{eq:lem_cova_Duv_cap_Du'v'}
\end{align}
Denote $\gamma_k \coloneqq \bigO\!\left(\frac{1}{\kb}\right)$. Therefore,
\begin{align*}
    &\PP[\Ec_{uv} \cap \Ec_{u'v'}] 
    = \left( 1 - \frac{1}{C_1 \sqrt{\kb}} \left(2\gamma 
       - \frac{\beta}{C_1 \sqrt{\kb}} + \gamma_k \right)
       \right)^{C_2 \sqrt{\kb}} \,,
\end{align*}
and substituting into~\eqref{eq:cova_formula} gives
\begin{align}
    &\Cov[\rE_{uv}, \rE_{u'v'}] \nonumber \\
    &= \left( 1 - \frac{1}{C_1 \sqrt{\kb}} \left(2\gamma 
       - \frac{\beta}{C_1 \sqrt{\kb}} + \gamma_k \right) 
       \right)^{C_2 \sqrt{\kb}} \nonumber \\
    &\qquad - \left( 1 - \frac{\gamma}{C_1 \sqrt{\kb}} 
       \right)^{2 C_2 \sqrt{\kb}} \nonumber \\
    &= \exp \left(-\frac{1}{C_1 \sqrt{\kb}} \left(2\gamma 
       - \frac{\beta}{C_1 \sqrt{\kb}} + \gamma_k \right) 
       + \gamma_k \right)^{C_2 \sqrt{\kb}} \nonumber \\
    &\qquad - \exp \left(-\frac{\gamma}{C_1 \sqrt{\kb}} 
       + \gamma_k \right)^{2C_2 \sqrt{\kb}} \nonumber \\
    &= \exp \left(-\frac{C_2}{C_1} \left(2\gamma - \frac{\beta}{C_1 \sqrt{\kb}} 
       + \gamma_k \right) 
       + \bigO \left(\frac{1}{\sqrt{\kb}} \right) \right) \nonumber \\
    &\qquad - \exp \left(-\frac{2C_2 \gamma}{C_1} 
       + \bigO \left( \frac{1}{\sqrt{\kb}}\right) \right) \nonumber \\
    &\leq \exp \left(-2C_1 \gamma 
       + \bigO \left(\frac{1}{\sqrt{\kb}}\right) \right) 
       \cdot \left(\exp \left( \frac{\beta}{\sqrt{\kb}}\right) - 1 \right) 
       \nonumber \\
    &\leq \bigO\left( \frac{1}{\sqrt{\kb}} \right), \label{eq:cov_final}
\end{align}
where we used the approximation $1 - x = \exp(-x + \bigO(x^2))$ as $x \to 0$, 
the inequality $\exp(x) \leq 1 + 2x$ for $x \in [0,1]$, the fact that 
$\gamma, \beta \in [0,1]$, and the choice $C_2 = C_1^2$.

		\item \textbf{$u = u'$ but $u \neq v$, $u \neq v'$.} Similar to~\eqref{eq:lem_cova_Dc_uv_cap} in Case 1, we have
		\begin{align*}
			&\PP[\barD_{uv} \cap \barD_{uv'}] \\
            &= \frac{1}{C_1^2 \kb} \PP\left[ \barD_{uv} \cap \barD_{uv'} \,\middle|\, h(u) = h(v) = h(v') \right] \,.
		\end{align*}
		Let $\gamma = \PP[\Dc_{uv} \mid h(u) = h(v)]$ and $\beta = \PP\left[ \barD_{uv} \cap \barD_{uv'} \mid h(u) = h(v) = h(v') \right]$. Then, similar to~\eqref{eq:lem_cova_Duv_cap_Du'v'}, we have  
        \begin{align*}
    \PP[\Dc_{uv} \cap \Dc_{uv'}] 
    &= \PP[\Dc_{uv}] + \PP[\Dc_{uv'}] - \PP[\Dc_{uv} \cup \Dc_{uv'}] \\
    &= 2 \left(1 - \frac{\gamma}{C_1 \sqrt{\kb}} \right) 
       - 1 + \frac{\beta}{C_1^2 \kb} \\
    &= 1 - \frac{1}{C_1 \sqrt{\kb}} \left( 2\gamma 
       - \frac{\beta}{C_1 \sqrt{\kb}} \right) \,,
    \end{align*}
        which implies
\begin{align*}
    \PP[\Ec_{uv} \cap \Ec_{uv'}] 
    &= \left( 1 - \frac{1}{C_1 \sqrt{\kb}} \left( 2\gamma 
       - \frac{\beta}{C_1 \sqrt{\kb}} \right) \right)^{C_2 \sqrt{\kb}} \,.
\end{align*}
Therefore, the same reasoning as \eqref{eq:cov_final} gives
\begin{align*}
    &\Cov[\rE_{uv}, \rE_{uv'}] \\
    &= \left( 1 - \frac{1}{C_1 \sqrt{\kb}} \left( 2\gamma 
       - \frac{\beta}{C_1 \sqrt{\kb}} \right) \right)^{C_2 \sqrt{\kb}} \\
    &\qquad - \left( 1 - \frac{\gamma}{C_1 \sqrt{\kb}} 
       \right)^{2 C_2 \sqrt{\kb}} \\
    &\leq \exp\left( - 2 C_1 \gamma 
       + \bigO \left(\frac{1}{\sqrt{\kb}}\right) \right) 
       \left( \exp\left( \frac{\beta}{\sqrt{\kb}} \right) - 1 \right) \\
    &\leq \bigO\left( \frac{1}{\sqrt{\kb}} \right) \,.
\end{align*}
	\end{enumerate}
	
    Finally, summing over all non-defective pairs gives 
\begin{align*}
    \Var_{\ell}\left[ \sum_{u,v} \rE_{uv} \right] 
    &= \sum_{u,v} \Var_{\ell} \left[\rE_{uv} \right] 
       + \sum_{u,v, u',v'} \Cov \left[ \rE_{uv}, \rE_{u'v'} \right] \\
    &= \bigO(\rEm) + \bigO \left( \frac{\rEm^2}{\sqrt{\kb}}\right) 
    = \bigO \left( \frac{\rEm^2}{\sqrt{\kb}}\right) \, ,
\end{align*}
since $\Var_{\ell}[\rE_{uv}] \leq \PP[\Ec_{uv}] \leq \exp(-(C_1 - 23))$ 
and there are at most $12\rEm$ such pairs by assumption, with 
$\rEm = \Omega(\kb)$ (see~\eqref{eq:number_edges_level}). 
This completes the proof.
\end{proof}


\begin{lemma}
\label{lem: bounds_size_PD}
For \( C_2 = C_1^2 \) and \( C_1 \geq 27 \), conditioned on the \( \ell \)-th level having at most \( 12\rEm \) possibly defective (PD) pairs, the number of PD pairs at the \( (\ell+1) \)-th level is at most \( 12\rEm \) with probability \( 1 - \bigO \big( \tfrac{1}{\sqrt{\kb}} \big) \).
\end{lemma}

\begin{proof}
Among the PD pairs at the \( \ell \)-th level, at most \( \rEm \) are defective pairs, which generate at most \( 6\rEm \) children at the next level (each PD pair produces four cross-pairs between the two blocks and two intra-block pairs).  
By Lemma~\ref{lem:expectation} and Lemma~\ref{lem:variance}, and applying Chebyshev's inequality, with probability at least \( 1 - \bigO \big( \tfrac{1}{\sqrt{\kb}}\big) \), at most \( \rEm \) non-defective pairs are incorrectly retained as PD.  
These contribute at most another \( 6\rEm \) children at the next level, leading to a total of at most \( 12\rEm \) PD pairs.
\end{proof}

We now prove the main theorem.
\begin{proof}[Proof of Theorem~\ref{thm:main}] The stated claims are inferred as follows.
 As shown in Lemma~\ref{lem:typical}, a random graph from $\ER(n,q)$ lies in the typical set of graphs with probability $1-o(1)$. Therefore, it suffices to analyze the  decoding time and error probability conditioned on an arbitrary typical graph $G \in \Tc(\epsilon_n)$, where $\epsilon_n = o(1)$.
\begin{itemize}
     
	\item \textbf{Decoding time:} 
		For the case \( \theta > \frac{1}{2} \), from~\eqref{eq:number_edges_level} used in Lemma~\ref{lem:typical}, there are at most \(\rEm =  4\kb \) defective pairs at each level. From Lemma~\ref{lem: bounds_size_PD} and by induction, for any given level \( \ell \), we have \( |\mathcal{PD}^{(\ell)}| \leq 12\rEm =   48 \kb\) with conditional probability at least \( 1 - \bigO \big(\frac{1}{\sqrt{\kb}} \big) \). Taking a union bound over \( \log_2 n \) levels, the same bound holds for all levels simultaneously with probability at least \( 1 - \bigO \big( \frac{\log n}{\sqrt{\kb}} \big) \). 
		
		The decoding time is dominated by the outcome checks in the decoding procedure. At each level \( \ell \), for each possible defective pair, we conduct at most \( C_1 \sqrt{\kb} \) outcome checks. This gives a total of \( \bigO \big( \sqrt{\kb} |\mathcal{PD}^{(\ell)}| \big) \) outcome checks at level \( \ell \). Therefore, the total number of outcome checks from levels \( \ceil{\log_2 \sqrt{\kb}} \) to \( \log_2 n - 1 \) is 
		\[
		\bigO \left( \sqrt{\kb} \sum_{\ell = \ceil{\log_2 \sqrt{\kb}}}^{\log_2 n - 1} |\mathcal{PD}^{(\ell)}| \right).
		\]
		At the final level \( \ell  = \log_2 n \), we conduct \( \bigO \big( \sqrt{\kb} \log n \cdot |\mathcal{PD}^{(\ell)}| \big) \) outcome checks.
		
		As shown above, with probability at least \( 1 - \bigO \big( \frac{\log n}{\sqrt{\kb}} \big) \), we have \( |\mathcal{PD}^{(\ell)}| \leq 48\kb \) for all \( \ell \in \{\ceil{\log_2 \sqrt{\kb}}, \dots, \log_2 n\} \). Therefore, the total number of outcome checks in the decoding procedure is at most \( \bigO(\kb^{1.5} \log n) \), with probability at least \( 1 - \bigO \big( \frac{\log n}{\sqrt{\kb}} \big) \).
		
		A similar argument applies to the case \( \theta \leq \frac{1}{2} \), with at most \( \rEm =  2\kb \log^2 \kb \) (see~\eqref{eq:number_edges_level}) defective pairs at each level.

    \item \textbf{Decoding time with pre-processing:} At each level $\ell$, for each pair $(u,v) \in [g]\times[g]$ with $g=2^\ell$, let $\ID^{\ell}(u,v)$ be the set of test indices that contain both $u$ and $v$. Recall that in each round (among $C_2\sqrt{\kb}$ rounds), the probability that $u$ and $v$ appear in the same test is $\frac{1}{C_1\sqrt{\kb}}$. The $C_2\sqrt{\kb}$ iterations are independent, and using $C_2=C_1^2$, we have
     \[
       \EE\!\big[\,|\ID^{\ell}(u,v)|\,\big] \;=\; \frac{C_2\sqrt{\kb}}{C_1\sqrt{\kb}} \;=\; C_1.
     \]
Therefore, by the Chernoff bound (see~\eqref{eq:chernoff2} in Appendix~\ref{appendix:concentration}), we have 
\begin{align*}
    &\PP\big[\,|\ID^{\ell}(u,v)| > 3\log n\,\big] \\
    &\;\le\; \left(\frac{eC_1}{3\log n}\right)^{3\log n} \\
    &\;=\; \bigO\!\left(\frac{1}{\log n}\right)^{\log n}.
\end{align*}
Taking a union bound over at most $\log_2 n$ levels, each with at most $n^2$ pairs, we obtain that for all $\ell \in \{\lceil \log_2 \sqrt{\kb}\rceil,\dots,\log_2 n\}$ and all $(u,v)\in[g]\times[g]$ with $g=2^\ell$, we have
\begin{align}
    &\PP\big[\,|\ID^{\ell}(u,v)| \le 3\log n\,\big] \notag \\
    &\;\ge\; 1 - \bigO\!\left(n^2\log n 
       \cdot \left(\frac{1}{\log n}\right)^{\log n}\right).
    \label{eq:bound_index_tests}
\end{align}
which in turn behaves as $1-o(1)$.

After fixing the test design, at each level $\ell$ and for each pair $(u,v)\in[g]\times[g]$, we can find $\ID^{\ell}(u,v)$ in $\bigO(\sqrt{\kb})$ time, since each block appears in exactly $C_2\sqrt{\kb}$ tests at level $\ell$. Hence, the total time to compute all sets $\ID^{\ell}(u,v)$ over all $\bigO(\log n)$ levels and all $\bigO(n^2)$ pairs is $\bigO(n^2\sqrt{\kb}\log n)$.

With this pre-processing, for every potentially defective pair at each level, the decoder only needs to check at most $3\log n$ test outcomes, rather than $\bigO(\sqrt{\kb})$ of them as we did before. Therefore, the scaling of the decoding time decreases by a factor of $\frac{\sqrt{\kb}}{\log n}$, meaning it becomes $\bigO(\kb\log^2 n)$ for $\theta>1/2$, and $\bigO(\kb\log^2\kb\,\log^2 n)$ for $\theta\le 1/2$.

\item{ \textbf{Error probability:} As shown above (in the analysis of decoding time), we have
\begin{align}
    &\PP\Big[ \max_{\ell_{\min} \le \ell \le \log_2 n} 
       \big|\mathcal{PD}^{(\ell)}\big| \le 12\rEm \Big] \notag \\
    &\;\ge\; 1 - \bigO \left( \frac{\log n}{\sqrt{k}} \right) \, .
    \label{eq:bound_PD_all_levels}
\end{align}
where we recall that $\rEm = 4\kb$ for $\theta>1/2$ and $\rEm = 2\kb \log^2 \kb$ for $\theta\le 1/2$.

At the final level, we conduct $C'\log n$ independent rounds, each consisting of $C_1 C_2 \kb$ tests. We analyze the error probability under the high-probability event that $\big|\mathcal{PD}^{(\log_2 n)}\big|\le 12\rEm$.

For any fixed non-defective pair $(u,v)$ in $\mathcal{PD}^{(\log_2 n)}$ at the final level, over a given sequence of $C_1 C_2 \kb$ tests, following the argument of~\eqref{eq:bound_pair_not_identify} in Lemma~\ref{lem:expectation}, we have
\begin{align*}
    &\PP\big[ (u,v) \text{ is not identified among } 
       C_1 C_2\kb \text{ tests} \big] \\
    &\;\le\; \exp\big( -(C_1-23) \big) \, .
\end{align*}
Since we perform $C'\log n$ repetitions with $C'>3$ independent rounds at the final level, it follows that whenever $C_1 \ge 24$, we have
\begin{align*}
    &\PP\big[ (u,v) \text{ is not identified at the final level} \big] \\
    &\le\; \bigO(n^{-C'}) \, .
\end{align*}

Applying a union bound over the $\big|\mathcal{PD}^{(\log_2 n)}\big|$ non-defective pairs at the final level, we obtain
\begin{align}
    &\PP \big[\, \widehat{E} = E \;\bigm|\; G \in \Tc(\epsilon_n),\ 
       \big|\mathcal{PD}^{(\log_2 n)}\big| \le 12\rEm \big] \notag \\
    &\;\ge\; 1 - \bigO \left( \frac{\rEm}{n^{C'}} \right) \, .
    \label{eq:error_prob}
\end{align}

Combining Lemma~\ref{lem:typical} with~\eqref{eq:bound_PD_all_levels}, for \(G \sim \ER(n,q)\) the estimated edge set \(\widehat{E}\) output by Algorithm~\ref{alg:decoding} matches the true edge set \(E\) with probability at least
\[
1 - o(1) - \bigO \left(\frac{\rEm}{n^{C'}}\right)
= 1 - o(1) \, \text{for any } C'>3.
\]
In other words, \(P_e \coloneqq \PP[\widehat{E} \ne E] \to 0\) as \(n \to \infty\).

}
\item { \textbf{Number of tests: }We used \( C_1 C_2 \kb \) tests at each level from \( \ceil{\log_2 \sqrt{\kb}} \) to \( \log_2 n - 1 \). At the final level, we used \( C_1 C_2 C' \kb \log n \) tests. Adding these leads to a total of  \( \bigO(\kb \log n) \).
	}


\end{itemize}

\end{proof}

\section{Improved Decoding Time via Partitioning and Permutations}
\label{sec:perm-screening-base}

The binary splitting approach in the previous section ensures recovery with an order–optimal number of tests, but the $\kb^{1.5}$ dependence in the decoding time (in the absence of pre-processing) has room for improvement. The bottleneck arises because the decoder needs to check \(\bigO(\sqrt{\kb})\) tests for each pair in the set of possible defectives. To address this, we will partition the vertex set \(V\) into equal-size subsets \(\Sc_1,\dots,\Sc_m\) and apply the binary-splitting approach to each induced subgraph formed by taking the union of pairs of these, i.e., \(G_{ij}=G[\Sc_i\cup \Sc_j]\). On these smaller graphs, the number of edges is much smaller, and thus the per-pair checking time in decoding is significantly reduced. 

The general idea of solving multiple smaller problems has been explored in non-adaptive group testing~\cite{li2024noisy}.  In short, their idea is to create multiple subproblems, and randomly place each item in a constant number of subproblems and ``zero out'' the rest, creating a problem with very few defectives, which is solved using binary splitting.  They decode each item via a combined ``vote'' over the (few) subproblems it was placed in.  However, there are at least two major issues in applying their approach to our setting: (i) It is difficult for us to use \emph{random} placements into subproblems while ensuring that \emph{pairs} of items (rather than just individual items) appear together in sufficiently many of them; and (ii) the error probability of each subproblem is harder to control tightly in our setting, and this lack of tightness can lead to a total exceeding 1 upon applying a union bound over subproblems.



To alleviate such issues, we incorporate a small collection of random permutations into the design: Before applying the non-adaptive tests on each $G_{ij}$, we permute its vertices several times and run the binary splitting algorithm under each permutation. This technique of \emph{permutation amplification} has been employed in various settings to boost the probability of success in randomized algorithms~\cite{motwani1996randomized}. Here it ensures that at least one permutation yields a favorable block structure for each $G_{ij}$, so that a union bound across all subgraphs guarantees high overall success probability. Crucially, the total number of tests remains $\bigO(\kb\log n)$.

We begin by formalizing the subgraph approach and introducing the permutation--based block construction. We then establish concentration bounds for edge counts and block degrees under random permutations. Next, we define the notion of a typical set of graphs with permutations and prove that such typicality holds with high probability. Finally, we present the permutation--based testing and decoding procedures and conclude with the main theorem on its guarantees. 

The notation used throughout this section will be summarized in Table \ref{tab:sec4-notation} below.

\subsection{Subgraphs and Random Permutations} \label{sec:subgraphs}

\textbf{Definitions of subgraphs.} Let \(G\sim \ER(n,q)\) with \(q=\Theta \big(n^{-2(1-\theta)}\big)\) for some \(\theta\in(0,1)\), and let \(\kb= q\binom{n}{2} =\Theta(n^{2\theta})\). 
Throughout this section, we fix a constant
\[
\gamma\in\Bigl(0,\min\Bigl\{1,\frac{1-\theta}{3\theta}\Bigr\}\Bigr)
\]
that will play an important rule in our design and analysis, and note that all asymptotic notation is with respect to \(n \to \infty\) (with \(\kb \to \infty\) simultaneously).
We partition \(V\) into
\[
m \coloneqq \sqrt{\frac{\kb}{\kb^\gamma}} = \kb^{\frac{1-\gamma}{2}}
\]
equal-size subsets \(\mathcal{S}_1,\dots,\mathcal{S}_m\)
. For each \((i,j) \in [m] \times [m] \) with \(i<j\), let \(G_{ij}\) be the subgraph induced by \(\mathcal{S}_i\cup \mathcal{S}_j\), and set
\begin{align*}
  k_{ij} \coloneqq\ |E(G_{ij})| \, , \text{ and } 
  n_{ij} \coloneqq\ |V(G_{ij})| = 2n \cdot \kb^{\frac{\gamma-1}{2}} \, .
\end{align*}
Let \(\kb_{ij} \coloneqq q\binom{n_{ij}}{2} = \Theta(\kb^{\gamma})\) denote the expected number of edges in \(G_{ij}\). Similar to Section~\ref{sec:assumption}, we assume (w.l.o.g.) that \(n_{ij}\) is a power of two for all \((i,j)\in[m] \times [m]\).  

We again use the idea of partitioning in levels (Section~\ref{sec:level_graph}).  We denote the index of the smallest level by $\ell_0$ (and its number of groups by $g_0$), and choose it slightly differently from before:
\begin{align}
\ell_0\coloneqq \ceil{\log_2 \kb^{2\gamma}}, \text{ and } g_0\ \coloneqq\ 2^{\ell_0} \ \le\ n_{ij}, \label{eq:new_l0}
\end{align}
where \(g_0\le n_{ij}\) follows from \(\gamma\le (1-\theta)/(3\theta)\).  This choice is slightly higher in the sense that Section \ref{sec:binary_spliting} would suggest using $\gamma/2$ instead of $2\gamma$ in the exponent.  This is done because one of our typical graph properties (see Lemma \ref{lem:empty-level} below, and the resultant condition (C2) in Definition \ref{def:new_typical}) will rely on $\frac{\kb_{ij}}{g_0} \to 0$.

For each pair \((i,j)\) and level \(\ell \geq \ell_0\), let \(g \coloneqq 2^{\ell}\) with \(g_0 \le g \le n_{ij}\).
Following Section~\ref{sec:typical_graphs}, we partition \(V(G_{ij})\) into \(g\) balanced blocks
\(\mathcal{U}\coloneqq\{U_1,U_2,\dots,U_g\}\), each of size \(n_{ij}/g\), where
\[
  U_v \coloneqq \big\{\, (v-1)\tfrac{n_{ij}}{g}+1,\ \dots,\ v\tfrac{n_{ij}}{g} \,\big\},
  \qquad v \in [g].
\]

\textbf{Pairwise independent permutation family.} Let $N \ge 1$ be an integer. A \emph{permutation} $\pi$ on $[N]$ is a bijection $\pi:[N]\to[N]$. Let $\mathcal{S}_N$ denote the set of all permutations on $[N]$. We say that $\mathcal{F}_N \subseteq \mathcal{S}_N$ is a \emph{pairwise independent permutation family} if, when a random permutation $\pi$ is drawn \emph{uniformly at random} from $\mathcal{F}_N$, the following holds for every distinct $x_1,x_2\in[N]$ and every distinct $y_1,y_2\in[N]$:
\[
  \PP_{\pi\sim \mathrm{Unif}(\mathcal{F}_N)} \big[\pi(x_1)=y_1 \, , \, \pi(x_2)=y_2\big]
  = \frac{1}{N(N-1)}.
\]
That is, the pair $(\pi(x_{1}),\pi(x_{2}))$ is uniformly distributed over  $\binom{N}{2}$ pairs.

When working with a family \(\mathcal{F}_N \subseteq \mathcal{S}_N\), explicitly storing an arbitrary permutation from \(\mathcal{F}_N\) requires \(N\) entries (i.e., \(\Theta(N\log N)\) bits). However, it is possible to design $\mathcal{F}_N$ so that each permutation is described by only a constant number of parameters, yielding an $\bigO(\log N)$-bit representation while still retaining the desired randomness properties (e.g., pairwise independence). In particular, if $N$ is a power of two, one can construct such a family $\mathcal{F}_N$ using affine permutations over the finite field of order $N$, $\mathbb{F}_N$. For $a\in\mathbb{F}_N \setminus \{0\}$ and $b\in\mathbb{F}_N$, define
\[
\pi_{a,b}(x)=ax+b,
\]
where all operations are over $\mathbb{F}_N$. Then each $\pi_{a,b}$ is a permutation of $\mathbb{F}_N$. Fix a bijection between $[N]$ and $\mathbb{F}_N$, which we write as  $[N]\cong\mathbb{F}_N$. Under this bijection, the family $\mathcal{F}_N$ is parameterized by $(a,b)\in (\mathbb{F}_N \setminus \{0\})\times \mathbb{F}_N$.  The construction details are given in Appendix~\ref{appendix:pairwise_independent}. Most importantly for our purposes, (i) the set $\{\pi_{a,b}\}$ forms a pairwise independent permutation family, and (ii) for any input $x$, both $\pi_{a,b}(x)$ and the inverse $\pi_{a,b}^{-1}(x)$ can be computed in $\bigO(1)$ time in the word-RAM model.

To simplify notation, assume (for now) that the vertex set \(V(G_{ij})\) is labeled as \(\{1,2,\dots,n_{ij}\}\). Draw \(c\) i.i.d.\ permutations \(\pi_1,\dots,\pi_c\) from a pairwise independent permutation family on $[n_{ij}]$. For any fixed \(t \in [c]\), define the permuted block family
\[
  \mathcal{U}^{(t)} \coloneqq \big\{\pi_t(U_1),\dots,\pi_t(U_{g})\big\},
\]
where \(\pi_t(U_p) \coloneqq \{\pi_t(v): v \in U_p\}\) for all \(p \in [g]\).
Since \(\pi_t\) is a bijection, \(\mathcal{U}^{(t)}\) is again a balanced partition of \(V(G_{ij})\) into \(g\) equal-size blocks.




Let \(\pi\) be a generic permutation on $[n_{ij}]$. We define the \emph{level-\(\ell\) block graph} \(H_{ij}^{(\pi,\ell)}\) (with \(g \coloneqq 2^{\ell} \in [g_0,n_{ij}]\)) as follows: its vertex set \(V\big(H_{ij}^{(\pi,\ell)}\big)\) consists of the \(g\) level-\(\ell\) blocks (each treated as a single vertex), and its edge set is defined via
\begin{align}
    (r,s)\in E\big(H_{ij}^{(\pi,\ell)}\big) \notag 
    &\iff \exists\, x\in \pi(U_r),\ \exists\, y\in \pi(U_s) \notag \\
    &\hspace{2em} \text{such that } (x,y)\in E(G_{ij}).
    \label{eq:H_edge_set}
\end{align}
Analogous to Section \ref{sec:level_graph}, we introduce the following definitions:
\begin{itemize}
    \item \(\PV_{g}\big(H_{ij}^{(\pi,\ell)}\big)\): the number of \emph{defective} blocks (vertices) in \(H_{ij}^{(\pi,\ell)}\), i.e., blocks having an internal edge.
    \item $d_{g}\big(H_{ij}^{(\pi,\ell)}\big)$: the maximum degree of $H_{ij}^{(\pi,\ell)}$. Here, for convenience, the maximum is taken over all blocks at level~$\ell$, rather than only the non-defective ones as in Section~\ref{sec:level_graph}, since we will ultimately only rely on block graphs that contain no defective blocks at any level. 
\end{itemize}

\begin{figure}[t]
\centering
\includegraphics[width=0.48\textwidth]{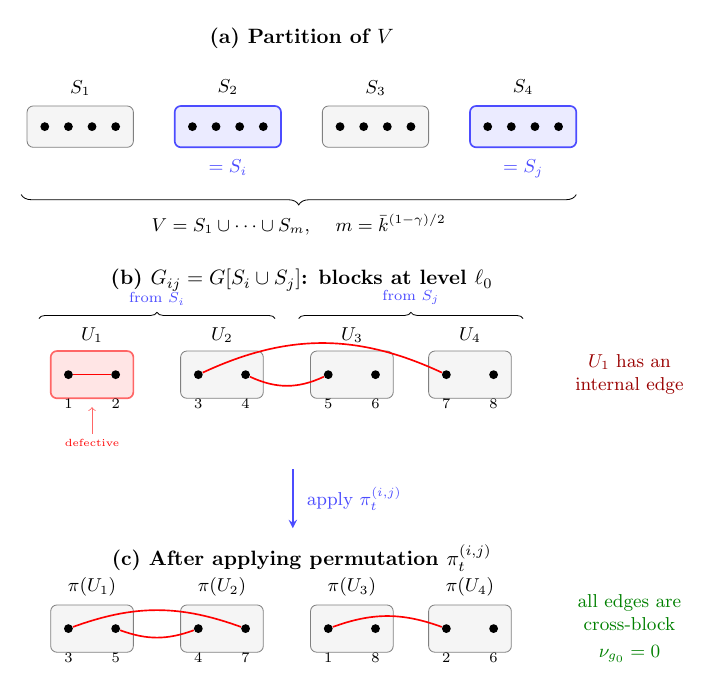}
\caption{Example illustrating the partitioning and permutation approach. (a)~The vertex set $V$ is partitioned into $m = \kb^{(1-\gamma)/2}$ balanced parts; the pair $(S_i, S_j)$ is highlighted. (b)~The induced subgraph $G_{ij} = G[S_i \cup S_j]$ is divided into $g_0 = 4$ blocks at the base level~$\ell_0$. Block $U_1$ contains an internal edge $(1,2)$ and is therefore defective ($\nu_{g_0} \geq 1$). The remaining edges $(3,7)$ and $(4,5)$ are cross-block. (c)~After applying a random permutation $\pi_t^{(i,j)}$, the vertices are redistributed among the blocks: $\pi(U_1) = \{3,5\}$, $\pi(U_2) = \{4,7\}$, $\pi(U_3) = \{1,8\}$, $\pi(U_4) = \{2,6\}$. Under this permutation, all three edges become cross-block and no block contains an internal edge, yielding $\nu_{g_0} = 0$. The binary splitting decoder of Section~\ref{sec:binary_spliting} can then be applied to $G_{ij}$ under this favorable block structure.}
\label{fig:partition-permutation}
\end{figure}

\subsection{High-Probability Graph and Permutation Properties}

In the following lemmas, we establish high-probability properties of the random subgraphs and random permutations, some of which serve as counterparts to those used in Section \ref{sec:binary_spliting}.  Note that we sometimes study the randomness of edges (for a fixed permutation) and sometimes study the randomness of permutations (for a fixed graph); we will accordingly take care when combining these properties in Section \ref{sec:typical_graph_permuation}.

\begin{lemma}
	\label{lem:number_edges_reducing}
	Under the preceding setup, with probability $1 - o(1)$, we have
	\begin{align*}  
		\kb^{\gamma} \leq k_{ij} \coloneqq |E(G_{ij})| \leq 12\kb^\gamma, \quad \forall (i,j) \in [m] \times [m], 
	\end{align*}
	where $m = \kb^{\frac{1-\gamma}{2}}$. 
\end{lemma}

\begin{proof}
	Each $G_{ij}$ has $n_{ij}$ vertices, so
	\[
	\kb_{ij} \coloneqq \mathbb{E}[k_{ij}] = q \cdot \binom{n_{ij}}{2} \, ,
	\]
	and recalling $n_{ij} =2n \cdot \kb^{\frac{\gamma-1}{2}}$, this implies
	\[ 
	  q \cdot n^2 \cdot \kb^{\gamma-1}  \leq  \kb_{ij} \leq 2q \cdot n^2 \cdot \kb^{\gamma-1} \, .
	\]
	Substituting $ \kb = q \binom{n}{2} \in \big[\frac{1}{4}qn^2,\frac{1}{2}qn^2 \big]$, we obtain
	\begin{equation}
	   2\kb^{\gamma} \leq \kb_{ij} \leq 8\kb^{\gamma}. \label{eq:k_ij_bounds}
	\end{equation}
    
	The edges are independently generated, so by the Chernoff bound, we have
	\[
	\PP\left[ |k_{ij} - \kb_{ij}| \geq  \frac{1}{2} \kb_{ij} \right]
	\leq \exp \big( - \Omega(\kb^{\gamma}) \big)  \, .
	\]
	Therefore, for each $(i,j)$, we have
    \begin{equation}
       \label{eq:bound_edge_component}
       \PP[\kb^{\gamma} \leq k_{ij} \leq 12\kb^{\gamma}] \geq 1- \exp\big( -\Omega(\kb^{\gamma})\big) \,.
    \end{equation}
	There are at most $ \binom{m}{2} \leq m^2 =  \kb^{1-\gamma}$ pairs, so by union bound, we have 
    \begin{align*}
    \PP\left[ \exists (i,j): k_{ij} > 12\kb^{\gamma} 
       \,\text{ or }\, k_{ij}< \kb^{\gamma} \right] 
    &\;\le\; m^2 \exp\big(-\Omega(\kb^{\gamma})\big) \,.
\end{align*}
	Therefore, with probability $1 - o(1)$, we have $\kb^{\gamma} \leq k_{ij} \leq 12\kb^{\gamma}$ for all $(i,j)$.
\end{proof}



Recall that $\Tc(\epsilon_n)$ denotes the typical set of graphs defined in Section~\ref{sec:typical_graphs}, and $\PV_g\big(H_{ij}^{(\pi,\ell)}\big)$ denotes the number of defective vertices (blocks) in $H_{ij}^{(\pi,\ell)}$, with $g = 2^{\ell}$. The following lemma shows that, for any fixed graph $G \in \Tc(\epsilon_n)$ and any fixed pair $(i,j) \in [m]\times[m]$, if \(\pi\) is chosen uniformly at random from a pairwise independent permutation family, then with high probability the block graph $H_{ij}^{(\pi,\ell_0)}$ contains no defective vertices, where $\ell_0 = \ceil{\log_2 \kb^{2\gamma}}$.  

\begin{lemma}\label{lem:empty-level}
Fix $G = (V,E)\in \Tc(\epsilon_n)$ and a pair \((i,j) \in [m] \times [m]\). 
Let \(\pi\) be drawn uniformly from a pairwise independent permutation family on \([n_{ij}]\), and let \(H_{ij}^{(\pi, \ell_0)}\) denote the level–\(\ell_0\) block graph (see \eqref{eq:H_edge_set}), where \(\ell_0 \coloneqq \ceil{\log_2 \kb^{2\gamma}}\) and \(g_0\coloneqq 2^{\ell_0}\). Then
\begin{align*}
  \PP\Big[\PV_{g_0}\big(H_{ij}^{(\pi, \ell_0)}\big)=0 \,\Big|\, G_{ij}\Big] 
  &\ \ge\ 1 - \frac{k_{ij}}{g_0} \, ,
\end{align*}
where $k_{ij} = |E(G_{ij})|$.
Moreover, if \(\pi_1,\ldots,\pi_c\) are i.i.d.~permutations drawn uniformly from a pairwise independent permutation family on \( [n_{ij}]\), then
\begin{align*}
  \PP\Big[\exists\, t \in [c]:\ \PV_{g_0} \big(H_{ij}^{(\pi_t,\ell_0)}\big)=0 
  \,\Big|\, G_{ij}\Big] 
  \;\ge\; 1 - \left(\tfrac{k_{ij}}{g_0}\right)^{c}.
\end{align*}
\end{lemma}

\begin{proof}
Index the unpermuted blocks by \(v\in[g_0]\), and define the indicator variable that the permuted block \(\pi(U_v)\) is defective:
\begin{align*}
    D_v \coloneqq\ \mathbf{1}\Big\{\exists\, e=(x,y)\in E(G_{ij}) 
    \text{ s.t. } \{x,y\}\subset \pi(U_v)\Big\}.
\end{align*}
For a fixed edge \(e=(x,y)\in E(G_{ij})\), let
\begin{align*}
    X_{v,e} &\coloneqq \mathbf{1}\{\{x,y\}\subset \pi(U_v)\} \\
    &= \mathbf{1}\{\{\pi^{-1}(x), \pi^{-1}(y) \}\subset U_v\}.
\end{align*}
Since \(\pi\) is pairwise independent, the unordered pair \((\pi^{-1}(x),\pi^{-1}(y))\) is uniform over \(\binom{n_{ij}}{2}\) pairs. Hence,
\begin{align*}
  \PP\big[X_{v,e}=1\big]\ =\ \frac{\binom{|U_v|}{2}}{\binom{n_{ij}}{2}}
  \ \le\ \frac{\binom{n_{ij}/g_0}{2}}{\binom{n_{ij}}{2}}
  \ \le\ \frac{1}{g_0^{\,2}}.
\end{align*}
By definition, \(D_v\le \sum_{e\in E(G_{ij})} X_{v,e}\), and hence
\begin{align*}
  \EE[D_v]\ \le\ \sum_{e\in E(G_{ij})}\EE[X_{v,e}]\ \le\ k_{ij}\cdot g_0^{-2}.
\end{align*}
Summing over \(v\) and using linearity of expectation,
\begin{align*}
  \EE\left[\PV_{g_0}\big(H_{ij}^{(\pi,\ell_0)}\big)\right]\ =\ \sum_{v=1}^{g_0} \EE[D_v]\ \le\ g_0\cdot \frac{k_{ij}}{g_0^{2}}\ =\ \frac{k_{ij}}{g_0}.
\end{align*}
By Markov’s inequality, \(\PP[\PV_{g_0}\ge 1]\le k_{ij}/g_0\), which yields the first part of the lemma. The second part follows directly from the independence of \(\{\pi_t\}_{t=1}^{c}\) across $t$.
\end{proof}

The following lemma shows that, for any pair $(i,j) \in [m]\times[m]$, the maximum degree of $H_{ij}^{(\pi,\ell)}$ is at most $\sqrt{\kb_{ij}}$ with high probability, where we recall that $\kb_{ij} = q \binom{n_{ij}}{2} = \Theta(\kb^{\gamma})$.

\begin{lemma}\label{lem:fixed_per_degree}
Fix $\theta \in (0,1)$ and $(i,j) \in [m] \times [m]$. Let $G \sim \ER(n,q)$ for some $q = \Theta(n^{-2(1-\theta)})$, and let $\pi$ be a fixed permutation on $[n_{ij}]$. For each $\ell \in \{\ell_0,\ell_0+1,\dots,\log_2(n_{ij})\}$, let $g \coloneqq 2^\ell$, and let \(H_{ij}^{(\pi,\ell)}\) denote the level-\(\ell\) block graph (see \eqref{eq:H_edge_set}) and $\kb_{ij} \coloneqq q \binom{n_{ij}}{2} = \Theta(\kb^{\gamma})$. Under the preceding setup and definitions, we have
\[
\PP \Big[d_{g}\big(H_{ij}^{(\pi,\ell)}\big)>\sqrt{\kb_{ij}} \Big]\ \le\ n_{ij} \cdot \left( \frac{16e}{\kb^{3\gamma/2}}\right)^{\kb^{\gamma/2}}  \, .
\]
\end{lemma}

\begin{proof}
Fix a particular block $U$ (a vertex of $H_{ij}^{(\pi,\ell)}$), and let $Z \coloneqq d_{g}(U)$ denote its degree in $H_{ij}^{(\pi,\ell)}$.  For notational convenience, suppose that $U$ is the $g$-th block, and for each $i \in [g-1]$, let $Z_i$ be the indicator random variable for the event that there is at least one edge between the two blocks $U$ and $U_i$. By definition, 
\begin{align*}
    Z = \sum_{i=1}^{g-1} Z_i \, .
\end{align*}
Since each edge is generated independently, we have
\begin{align*}
    \PP[Z_i = 1] = 1 - (1-q)^{|U|^2} \, .
\end{align*}
Since $|U|=\tfrac{n_{ij}}{g}$, similar to~\eqref{eq:bound_prob_defective}, substituting $n_{ij}=2n\,\kb^{(\gamma-1)/2}$ and $q=\tfrac{\kb}{\binom{n}{2}}$, and using the standard bound $\tfrac{q|U|^{2}}{2}\le 1-(1-q)^{|U|^{2}}\le q|U|^{2}$ (for $0\le q|U|^{2}\le 1$), we obtain 
\begin{align*}
    \PP[Z_i=1] \;&\le\; \frac{16\kb^{\gamma}}{g^{2}},\\
    \EE[Z] \;&\le\; \frac{16\kb^{\gamma}}{g}\, .
\end{align*}

From $g\geq g_{0} \geq \kb^{2\gamma}$, it follows that
$\EE[Z]\le \frac{16\kb^{\gamma}}{g}\le  16 \kb^{-\gamma}$, and we also recall that $\sqrt{\kb_{ij}} \ge \kb^{\gamma}$ (see \eqref{eq:k_ij_bounds}). 
Combining these with the Chernoff bound (\eqref{eq:chernoff2} in Appendix~\ref{appendix:concentration}) gives
\begin{align*}
    \PP \Big[ Z \geq \sqrt{\kb_{ij}} \Big] \leq \Big( \frac{e\EE[Z]}{\sqrt{\kb_{ij}}}\Big)^{\sqrt{\kb_{ij}}}   
    \leq \Big( \frac{16e}{\kb^{3\gamma/2}} \Big)^{\kb^{\gamma/2}}.
\end{align*}
The proof is completed by applying a union bound over the $n_{ij}$ vertices of $G_{ij}$.
\end{proof}

\subsection{Typical Graphs Under Shared Permutations}
\label{sec:typical_graph_permuation}

Recall that \(\ell_0 \coloneqq \ceil{\log_2 \kb^{2\gamma}}\) and \(g_0 \coloneqq 2^{\ell_0}\). For each pair \((i,j) \in [m] \times [m]\), and fix a bijection (labeling map)
\[
  \lambda_{ij}: V(G_{ij}) \to [n_{ij}] \, .
\]
See \eqref{eq:fix_labeling} below for an explicit choice of $\lambda_{ij}$ that suffices for our purposes.

Let \(\Pi = (\pi_1,\ldots,\pi_c)\) denote \(c\) i.i.d.\ permutations drawn uniformly from a pairwise independent permutation family on \([n_{ij}]\).
For a given pair \((i,j)\) and \(t\in[c]\), we \emph{reuse} \(\pi_t\) on \(V(G_{ij})\) by conjugation with the labeling map:
\[
  \pi_t^{(i,j)} \coloneqq \lambda_{ij}^{-1}\circ \pi_t \circ \lambda_{ij} : V(G_{ij}) \to V(G_{ij}).
\]
Note that in Section~\ref{sec:subgraphs}, for notational convenience, we assumed that
\(V(G_{ij})\) is labeled as \(\{1,2,\dots,n_{ij}\}\), in which case \(\lambda_{ij}\) is the identity map.
We now explicitly introduce \(\lambda_{ij}\) in order to drop this assumption.

As in Section~\ref{sec:subgraphs}, we define a partition \(\Sc_1,\dots,\Sc_m\) of \(V\) by
\begin{equation}
     \label{eq:partition_V}
     \Sc_i = \{(i-1)\tfrac{n}{m}+1,\dots, i\tfrac{n}{m} \} \quad \forall i \in [m] \, .
\end{equation}
Then \(\Sc_1,\dots,\Sc_m\) is a partition of \(V\) into \(m\) parts of equal size.
Recall that \(G_{ij} = G[\Sc_i \cup \Sc_j]\) with \(i<j\) and
\(n_{ij} = |V(G_{ij})| = \tfrac{2n}{m}\).
With this construction, the map \(\lambda_{ij}\colon \Sc_i \cup \Sc_j \to [n_{ij}]\) is defined as follows.
For each \(x \in \Sc_i \cup \Sc_j\),
\begin{equation}
 \label{eq:fix_labeling}
    \lambda_{ij}(x) =
\begin{cases}
  p, & \text{if } x = (i-1)\dfrac{n}{m}+p,\\[2mm]
  p+\dfrac{n}{m},  & \text{if } x = (j-1)\dfrac{n}{m}+p \, .
\end{cases}
\end{equation}
In this construction, \(\lambda_{ij}\) is fully determined by \(i\) and \(j\), and both
\(\lambda_{ij}\) and \(\lambda_{ij}^{-1}\) can be computed in time \(\bigO(1)\).

For any level \(\ell\) with \(g = 2^{\ell} \in [g_0, n_{ij}]\), similar to~\eqref{eq:H_edge_set}, we let \(H_{ij}^{(t,\ell)}\) denote the 
level-\(\ell\) block graph obtained from \(G_{ij}\) after applying the permutation  \(\pi_t^{(i,j)}\). Specifically, its vertex set \(V(H_{ij}^{(t,\ell)})\) consists of the \(g\) level-\(\ell\) blocks (each treated as a single vertex), and its edge set is defined via:
\begin{align}
    (r,s)\in E(H_{ij}^{(t,\ell)}) \notag 
    &\iff \exists\, x\in \pi_{t}^{(i,j)}(U_r),\ 
       \exists\, y\in \pi_{t}^{(i,j)}(U_s) \notag \\
    &\hspace{2em} \text{such that } (x,y)\in E(G_{ij}).
    \label{eq:H_edge_set_t}
\end{align}
which matches \eqref{eq:H_edge_set} with $\pi = \pi_{t}^{(i,j)}$.  
In addition, \(\PV_g(\cdot)\) denotes the number of defective blocks at level \(\ell\) with \(g = 2^{\ell}\), and \(d_g(\cdot)\) denotes the maximum degree of the block graph $H_{ij}^{(t,\ell)}$.

\begin{definition} \label{def:new_typical}
    We define the \emph{typical set of graphs with permutations} \(\mathcal{T}(\epsilon_n,c,\gamma)\) to be the collection of all pairs \((G,\Pi)\) such that \(G \in \Tc(\epsilon_n)\) (see Definition \ref{def:typical}) and, for every \((i,j)\in[m]\times[m]\), the following conditions hold:
    \begin{enumerate}
      \item[(C1)] The number of edges in \(G_{ij}\), denoted \(k_{ij}\coloneqq |E(G_{ij})|\), satisfies \(\kb^{\gamma} \le k_{ij} \le 12\,\kb^{\gamma}\) and $k_{ij} \leq 2\kb_{ij}$, where we recall that $\kb_{ij} = q \binom{n_{ij}}{2}$ is the expected number of edges of $G_{ij}$.
      \item[(C2)] There exists \(t\in[c]\) such that \(\PV_{g_0} \big(H_{ij}^{(t,\ell_0)}\big)=0\).
      \item[(C3)] For every level \(\ell\) with \(g=2^{\ell}\in[g_0,n_{ij}]\) and every \(t\in[c]\), we have \(d_{g} \big(H_{ij}^{(t,\ell)}\big)\leq \sqrt{\kb_{ij}}\). 
    \end{enumerate}
\end{definition}

We combine Lemma~\ref{lem:typical}, Lemma~\ref{lem:number_edges_reducing}, Lemma~\ref{lem:empty-level}, and Lemma~\ref{lem:fixed_per_degree} to obtain the following lemma. 
\begin{lemma}\label{lem:typical_permutation}
Fix \(\theta \in (0,1)\) and \(\gamma \in \bigl(0, \min \{1, \tfrac{1-\theta}{3\theta}\bigr) \}\), and let \(c > 1/\gamma\). Let \(\Pi=(\pi_1,\ldots,\pi_c)\) denote \(c\) i.i.d.\ permutations drawn uniformly from a pairwise independent permutation family on \([n_{ij}]\) together with \(G \sim \mathrm{ER}(n,q)\), where $q = \Theta(n^{-2(1-\theta)})$. Under the preceding setup and definitions, there exists a nonnegative sequence \(\epsilon_n \to 0\) such that
\begin{align*}
  \PP \Big[ (G,\Pi) \in \mathcal{T}(\epsilon_n,c,\gamma) \Big]\ \to\ 1 
  \quad \text{as } n \to \infty.
\end{align*}
\end{lemma}


\begin{proof}
From Lemma~\ref{lem:typical}, with probability at least $1-o(1)$ we have $G \in \Tc(\epsilon_n)$. 

\medskip\noindent
\textbf{For (C1):} By Lemma~\ref{lem:number_edges_reducing}, with probability $1-o(1)$, for all $(i,j)$ we have
\[
  \kb^{\gamma} \le k_{ij} \le 12\,\kb^{\gamma}.
\]
Similar to~\eqref{eq:bound_number_edges}, we have
\[
\PP\big[k_{ij} > 2\kb_{ij}\big] \le \exp \big(-\Omega(\kb_{ij})\big) = \exp \big(-\Omega(\kb^{\gamma})\big) \, .
\]
Taking a union bound over $\binom{m}{2} \le m^2 = \kb^{1-\gamma}$ pairs, we obtain that, with probability at least $1-o(1)$, $k_{ij} \le 2\kb_{ij} \text{ for all }(i,j)\in[m]\times[m] \,$.

\medskip\noindent
\textbf{For (C2):} By Lemma~\ref{lem:empty-level}, 
\begin{align*}
    &\PP_{\Pi}\Big[\exists\,t\in[c]:\ 
       \PV_{g_0}\big(H^{(t,\ell_0)}_{ij}\big)=0\ \Big|\ G\Big] \\
    &\;\ge\; \max \left\{0,\ 1 - \left(\frac{k_{ij}(G)}{g_0}\right)^{c}\right\} \, ,
\end{align*}
where $k_{ij}(G)$ denotes the number of edges in $G_{ij}$. By the law of total probability,  
\begin{align*}
&\PP_{G,\Pi}\left[\exists\,t\in[c]:\, \PV_{g_0}\left(H^{(t,\ell_0)}_{ij}\right)=0\right] \\
&= \EE_{G}\left[\,\PP_{\Pi}\left(\exists\,t\in[c]:\, \PV_{g_0} \left(H^{(t,\ell_0)}_{ij}\right)=0 \,\middle|\, G\right)\right] \\
&\ge \EE_{G} \left[ \max\Big\{0,\ 1 - \Big(\tfrac{k_{ij}(G)}{g_0}\Big)^{c}\Big\} \right] \\
&\ge 1 - \bigO \left( \frac{1}{\kb^{c \gamma}} + \exp(-c'\kb^{\gamma})\right)\, ,
\end{align*}
where in the last inequality we used $g_0 \geq  \kb^{2\gamma}$ and~\eqref{eq:bound_edge_component} in Lemma~\ref{lem:number_edges_reducing}, 
which states that $\PP[\kb^{\gamma} \leq k_{ij}(G) \leq 12\kb^{\gamma}] \geq 1 - \exp \big(-\Omega(\kb^{\gamma}) \big)$. Since $c> 1/\gamma$, taking a union bound over the $ \binom{m}{2} \leq m^2 = \kb^{1-\gamma}$ pairs $(i,j)$ shows that, with probability at least $1-o(1)$, condition (C2) holds.

\medskip\noindent
\textbf{For (C3):} By Lemma~\ref{lem:fixed_per_degree}, for any level $\ell$ and permutation $\pi_t$ (applied to $G_{ij}$ via $\pi_t^{(i,j)}$), we have
\[
\PP_{G}\Big[\,d_g\big(H^{(t,\ell)}_{ij}\big)>\sqrt{\kb_{ij}}\ \Big|\ \pi_t\Big]
\ \le\ n_{ij}\left(\frac{16e}{\kb^{3\gamma/2}}\right)^{\kb^{\gamma}/2}.
\]
Applying a union bound over the $c$ permutations and at most $\log_2 n$ levels $\ell\in\{\ell_0,\ldots, \log_2 n_{ij}\}$ yields
\begin{align*}
    &\PP_{G} \Big[\,\exists\,t\in[c],\ \exists\,\ell\ge\ell_0:\ 
       d_g\big(H^{(t,\ell)}_{ij}\big)>\sqrt{\kb_{ij}}\ \Big|\ \Pi\Big] \\
    &\;\le\; c\log_2 n \cdot n_{ij} 
       \left(\frac{16e}{\kb^{3\gamma/2}}\right)^{\kb^{\gamma}/2}.
\end{align*}
By the law of total probability and $n_{ij} = 2n \cdot \kb^{\frac{\gamma-1}{2}}$, we have
\begin{align*}
    &\PP_{G,\Pi} \Big[\,\exists\,t,\ell:\ d_g \big(H^{(t,\ell)}_{ij}\big)>\sqrt{\kb_{ij}}\Big] \\
&=\EE_{\Pi} \Big[\ \PP_{G} \big(\exists t,\ell:\ d_g>\sqrt{\kb_{ij}}\ \big|\ \Pi\big)\ \Big] \\
 &\le\ c\log_2 n \cdot n_{ij} \left(\frac{16e}{\kb^{3\gamma/2}}\right)^{\kb^{\gamma}/2} \\
  &\le\ c\log_2 n \cdot 2n \cdot \kb^{\frac{\gamma-1}{2}} \left(\frac{16e}{\kb^{3\gamma/2}}\right)^{\kb^{\gamma}/2} \, .
\end{align*}
Taking a union bound over the $ \binom{m}{2} \leq m^2 = \kb^{1-\gamma}$ pairs $(i,j)$ shows that, with probability $1-o(1)$, condition (C3) holds.
\end{proof}



\begin{remark}
The same $\Pi$ is reused for all pairs $(i,j)\in[m]\times[m]$; differences across $(i,j)$ arise solely from the fixed label bijections $\lambda_{ij}$ as defined in~\eqref{eq:fix_labeling}.
\end{remark}

\subsection{Testing and Decoding Procedures}
\label{sec:testing-decoding_permutation}
We now present the testing and decoding procedures. For ease of reference, Table~\ref{tab:sec4-notation} summarizes the main
notation and constants used in this section.

\begin{table*}[t]
\centering
\renewcommand{\arraystretch}{1.1}
\begin{tabular}{|c|l|l|}
\hline
\textbf{Symbol} & \textbf{$\qquad$ Definition / Constraint} & \hspace*{1.4cm}\textbf{Description / Role} \\
\hline
\multicolumn{3}{|l|}{\emph{Partition and subgraph parameters}} \\
\hline
$\gamma$             & $\gamma \in \bigl(0,\, \min\{1,\, \tfrac{1-\theta}{3\theta}\}\bigr)$ & Controls partition size             \\
\hline
$m$                  & $m = \kb^{(1-\gamma)/2}$                                              & Number of subsets in partition            \\
\hline
$S_1,\dots,S_m$      & Balanced partition of $V$                                             & Vertex partition                  \\
\hline
$G_{ij}$             & $G[S_i \cup S_j]$                                                     & Induced subgraph                  \\
\hline
$n_{ij}$             & $2n \cdot \kb^{(\gamma-1)/2}$                                         & Number vertices of $G_{ij}$              \\
\hline
$\kb_{ij}$           & $q\binom{n_{ij}}{2} = \Theta(\kb^{\gamma})$                           & Expected number of edges in $G_{ij}$        \\
\hline
$\ell_0$             & $\lceil \log_2 \kb^{2\gamma} \rceil$                                  & Base level index                  \\
\hline
$g_0$                & $2^{\ell_0}$                                                           & Number of groups at base level              \\
\hline
\multicolumn{3}{|l|}{\emph{Permutation parameters}} \\
\hline
$c$                  & $c > 1/\gamma$                                                         & Number of permutations            \\
\hline
$\pi_1,\dots,\pi_c$  & i.i.d.\ from pairwise indep.\ family on $[n_{ij}]$                   & Global collection of permutations          \\
\hline
$\lambda_{ij}$       & $V(G_{ij}) \to [n_{ij}]$                                              & Labeling map                      \\
\hline
$\pi_t^{(i,j)}$      & $\lambda_{ij}^{-1} \circ \pi_t \circ \lambda_{ij}$                   & Permutation on $V(G_{ij})$        \\
\hline
\multicolumn{3}{|l|}{\emph{Block graph quantities}} \\
\hline
$H_{ij}^{(t,\ell)}$  & See Eq.~\eqref{eq:H_edge_set_t}                                                          & Block graph of $G_{ij}$ under $\pi_t^{(i,j)}$ at level $\ell$ \\
\hline
$\nu_g(\cdot)$       & Number of defective blocks                                             & Counts blocks with internal edges        \\
\hline
$d_g(\cdot)$         & Maximum degree of block graph                                          & Controls block graph regularity            \\
\hline
$\mathcal{PD}_{ij}^{(t,\ell)}$ & Possibly defective pairs at level $\ell$                      & Maintains candidate edge pairs              \\
\hline
\multicolumn{3}{|l|}{\emph{Testing and decoding constants}} \\
\hline
$C_1$                & $C_1 > 27$                                                             & Constant in tests per repetition      \\
\hline
$C_2$                & $C_2 = C_1^2$                                                          & Constant in number of repetitions     \\
\hline
$C_3$                & $C_3 \geq 3e$                                                          & Constant for base level testing       \\
\hline
$c'$                 & $c' > 2/\gamma$                                                        & Number of repetitions (amplification)\\
\hline
$C'$                 & $C' > \frac{8}{1 - \theta + \gamma\theta}$                             & Final level amplification         \\
\hline
\end{tabular}
\vspace{0.5em}
\caption{Summary of notation and constants in this section. Here $\theta \in (0,1)$ is the sparsity parameter, $\kb = q\binom{n}{2} = \Theta(n^{2\theta})$ is the expected number of edges, and $q = \Theta(n^{-2(1-\theta)})$ is the edge probability.}
\label{tab:sec4-notation}
\end{table*}

For each pair \((i,j)\), recall that \(G_{ij}\) is the induced subgraph on \(\mathcal{S}_i\cup \mathcal{S}_j\), and that we defined
\begin{align*}
  k_{ij}\ = |E(G_{ij})|, \qquad \ell_0\ = \ceil{\log_2 \kb^{2\gamma}}, \qquad g_0 = 2^{\ell_0} \, ,
\end{align*}
where \(\gamma\in\big(0, \min \{1, \tfrac{1-\theta}{3\theta} \}\big)\) and \(\kb = q \binom{n}{2} = \Theta(n^{2\theta})\). We draw, once and for all, a global tuple $\Pi=(\pi_1,\dots,\pi_c)$ of i.i.d.\ permutations sampled from a given pairwise independent permutation family; these are reused for \emph{all} pairs $(i,j)$ according to the labelings defined in~\eqref{eq:fix_labeling}. Let $G \sim \ER(n,q)$ with $q = \Theta \big(n^{-2(1-\theta)}\big)$ for $\theta \in (0,1)$.  By Lemma~\ref{lem:typical_permutation}, we have $(G,\Pi) \in \Tc(\epsilon_n,c,\gamma)$ with probability at least $1-o(1)$. In particular, for each pair $(i,j) \in [m]\times[m]$, $G_{ij}$ satisfies the following typical subgraph properties:
\begin{itemize}
    \item The number of edges in $G_{ij}$ satisfies $\kb^{\gamma} \le k_{ij}  \le 12\,\kb^{\gamma}$ and $k_{ij} \le 2\kb_{ij}$, where $\kb_{ij} = q\binom{n_{ij}}{2}$ with $n_{ij} =2n \cdot \kb^{\frac{\gamma-1}{2}}$.
    \item There exists $t\in [c]$ such that $\PV_{g_0} \big(H_{ij}^{(t,\ell_0)}\big)=0$. 
    This implies that for all levels $\ell \ge \ell_0$ we have $\PV_g \big(H_{ij}^{(t,\ell)}\big)=0$, i.e., there are no defective blocks in the block graphs $H_{ij}^{(t,\ell)}$ for $\ell \ge \ell_0$.
    \item For every level $\ell \ge \ell_0$ (with $g=2^{\ell}$) and $t \in [c]$, we have $d_g \big(H_{ij}^{(t,\ell)}\big) \le \sqrt{\kb_{ij}}$.
\end{itemize}
With these properties, we can apply the testing–decoding scheme of Algorithms~\ref{alg:testing} and~\ref{alg:decoding} to $G_{ij}$ using the permutation $\pi_t^{(i,j)}$, with the minor change of starting from a slightly higher initial level $\ell_0$ as discussed following \eqref{eq:new_l0}.  The details are described as follows.

We first describe the testing procedure, which is given in Algorithm~\ref{alg:testing_permutation}.  For each permutation \(\pi_t^{(i,j)}\) with \(t\in[c]\), we apply Algorithm~\ref{alg:testing} to \(G_{ij}\) with an expected number of edges \(\kb_{ij}\), from level \(\ell_0=\lceil \log_2 \kb^{2\gamma}\rceil\) up to \(\log_2 n_{ij}-1\), as follows.
We first perform $5C_3\,\kb_{ij}\log n$ tests on the nodes at level $\ell_0$; as we will see below, these will be used to find an index $t$ with $\PV_{g_0} \big(H_{ij}^{(t,\ell_0)}\big)=0$, where $g_0=2^{\ell_0}$.
Then, at each level $\ell$, as in Algorithm~\ref{alg:testing}, we randomly assign each vertex (block) of the block graph $H_{ij}^{(t,\ell)}$ to one of $C_1 \sqrt{\kb_{ij}}$ tests; i.e., if a vertex $v$ with corresponding block $U_v$ is chosen for a test, then all elements of $\pi_t^{(i,j)}(U_v)$ are placed into that test. 
This procedure is repeated for $C_2 \sqrt{\kb_{ij}}$ iterations. 
At the final level $\log_2 n_{ij}$, we perform $C' \log n_{ij}$ iterations, each consisting of $C_1 C_2 \kb_{ij}$ tests, as in Algorithm~\ref{alg:testing}. In addition, for each $G_{ij}$ and permutation $\pi_t^{(i,j)}$, we repeat the above procedure some number $c'$ of times independently to amplify the success probability.

\begin{algorithm*}[t]
\caption{Test Design Based on Permutations and Smaller Subproblems}
\label{alg:testing_permutation}
\begin{algorithmic}[1]
\REQUIRE
\parbox[t]{0.88\linewidth}{
Number of vertices $n$; number of expected edges $\kb$; parameter
$\gamma \in (0, \min {1, \frac{1-\theta}{3\theta} } )$; number of
permutations $c$; number of design repetitions $c'$; parameter
$C_3 \ge 3e$.
}

\STATE Partition $V$ into $m \gets \kb^{(1-\gamma)/2}$ balanced parts
$\Sc_1,\dots,\Sc_m$ as in~\eqref{eq:partition_V}.

\STATE For each $(i,j)\in[m]\times[m]$, set
$G_{ij} \gets G[\Sc_i \cup \Sc_j]$, \ $n_{ij} \gets |V(G_{ij})|$,
\ $\kb_{ij} \gets q \binom{n_{ij}}{2}$.

\STATE Set $\ell_{0} \gets \lceil \log_2 \kb^{2\gamma} \rceil$ and
$g_0 \gets 2^{\ell_0}$.

\STATE
\parbox[t]{0.93\linewidth}{
Draw \(c\) independent permutations \(\Pi=(\pi_1,\dots,\pi_c)\)
from a pairwise independent permutation family on \([n_{ij}]\);
fix relabelings $\lambda_{ij}: V(G_{ij}) \to  [n_{ij}]$ as
in~\eqref{eq:fix_labeling} and define
$\pi_t^{(i,j)} \gets \lambda_{ij}^{-1} \circ \pi_t \circ \lambda_{ij}$.
}

\FOR{each pair $(i,j)\in[m]\times[m]$}
  \FOR{$t=1,\dots,c$}
    \FOR{$r=1,\dots,c'$}
      \STATE
      \parbox[t]{0.91\linewidth}{
      \textit{\color{blue}// Level-$\ell_0$ tests to identify a permutation
      $t \in [c]$ s.t.~$H_{ij}^{(t,\ell_0)}$
      (see \eqref{eq:H_edge_set_t}) has no defective blocks}
      }

      \FOR{each iteration in \( \{1, \dotsc, 5 \log n \} \)}
        \STATE Initialize a sequence of $C_3 \kb_{ij}$ tests
        \FOR{each block  $U_{\kappa}$ with \( \kappa = 1, \dotsc, g_0 \)}
          \STATE Assign block \( \pi_{t}^{(i,j)}(U_{\kappa}) \) to a randomly chosen test among the $C_3 \kb_{ij}$ tests
        \ENDFOR
      \ENDFOR

      \STATE \textit{\color{blue}// Run Algorithm~\ref{alg:testing} on the smaller subgraph $G_{ij}$}

      \STATE
      \parbox[t]{0.91\linewidth}{
      Apply Algorithm~\ref{alg:testing} to $G_{ij}$, whose expected
      number of edges is $\kb_{ij}$, using the permutation
      $\pi_t^{(i,j)}$, from level $\ell_{0}$ up to $\log_2 n_{ij}$.
      }
    \ENDFOR
  \ENDFOR
\ENDFOR

\end{algorithmic}
\end{algorithm*}


Next, we describe the decoding procedure, which is given in Algorithm~\ref{alg:decoding_permutation}.  We rely on the high-probability event $(G,\Pi) \in \Tc(\epsilon_n,c,\gamma)$, which includes typical subgraph properties for each $G_{ij}$.  
First, we identify an index $t$ such that $\PV_{g_0} \big(H_{ij}^{(t,\ell_0)}\big)=0$ using a procedure described in Lemma~\ref{lem:find_permutation_base} below. 
For the resulting pair $(G_{ij},\pi_t^{(i,j)})$, note that $c'$ independent test designs have been prepared for $(G_{ij},\pi_t^{(i,j)})$.
We execute the decoding procedure of Algorithm~\ref{alg:decoding} on each of these designs, indexed by $r=1,2,\ldots,c'$.
In a given round $r\in[c']$, the decoder proceeds from level $\ell_0$ up to $\log_2 n_{ij}$ and continues only while the possibly defective set $\mathcal{PD}_{ij}^{(t,\ell)}$ remains suitably bounded, namely
\[
  \big|\mathcal{PD}_{ij}^{(t,\ell)}\big| \;\le\; 7\,\kb^{4\gamma}.
\]
If this condition is violated at any level, the current round $r$ is terminated, and the algorithm proceeds to $r+1$.
We consider a round $r$ to be \emph{successful} if decoding completes all levels up to $\log_2 n_{ij}$, and our analysis will only rely on having at least one such success per $(i,j)$ pair.

As we will show in Lemma~\ref{lem:pd_size_permutation_level} below, each round succeeds with probability $1-\bigO\Big(\tfrac{\log n_{ij}}{\kb^{\gamma/2}}\Big)$. 
By independence across the $c'$ designs, the probability that at least one round succeeds among the $c'$ rounds is $1 - \bigO \Big( \big(\tfrac{\log n_{ij}}{\kb^{\gamma/2}}\big)^{c'} \Big)$, which will be useful for applying a union bound over all $(i,j)$ pairs.

\begin{remark}
 Note that the decoding algorithm works on the vertex index set of the block graph $H_{ij}^{(t,\ell)}$, whose vertex set is $\{1,\dots,g\}$ with $g = 2^\ell$. Hence, the possible defective set $\mathcal{PD}_{ij}^{(t,\ell)}$ records only these block-pair indices. The actual edges in the original graph $G_{ij}$ are recovered by mapping each PD index $(r,s)\in [n_{ij}]\times[n_{ij}]$ back to its corresponding edge in $G_{ij}$ via the fixed labeling $\lambda_{ij}$ and the permutation $\pi_t^{(i,j)}$ used during test design. Specifically, when a block $U$ is added to a test, we actually include $\pi_t^{(i,j)}(U)$ in that test. Therefore, in the final level of decoding, we apply $\pi_t^{(i,j)} \circ \lambda_{ij}^{-1}$ to each PD index in order to recover the original edge endpoints in $G_{ij}$ (see Step~14 of Algorithm~\ref{alg:decoding_permutation}).

\end{remark}

\begin{algorithm*}[t]
\caption{Decoder Design Based on Permutations and Smaller Subproblems}
\label{alg:decoding_permutation}
\begin{algorithmic}[1]
\REQUIRE The fixed non-adaptive test designs and outcomes from Algorithm~\ref{alg:testing_permutation} at all levels $\ell \ge \ell_{0}$,
and the global tuple of permutations \(\Pi=(\pi_1,\dots,\pi_c)\) constructed in Algorithm~\ref{alg:testing_permutation}, where each \(\pi_t\) is drawn
from a pairwise independent permutation family.

\STATE Initialize the global edge estimate $\widehat{E} \gets \emptyset$.
\FOR{each pair $(i,j) \in [m] \times [m]$}
  \STATE \textit{\color{blue}// Permutation selection at the base level}
  \STATE Find $t^\star \in \{1,\dots,c\}$ such that $\PV_{g_0} \big(H_{ij}^{(t^\star,\ell_{0})}\big) = 0$, using the procedure described in Lemma~\ref{lem:find_permutation_base}.  \\
  
\textit{\color{blue}// Run Decoding Algorithm~\ref{alg:decoding} for $c'$ times (once per repetition)} \\
  \FOR{$r=1,\dots,c'$}
    \STATE Set $\ell \gets \ell_0$
    \WHILE{$\ell \le \log_2 n_{ij}-1$}
      \STATE Apply the decoding Algorithm~\ref{alg:decoding}'s update at level $\ell$ to obtain $\mathcal{PD}_{ij}^{(t^{\star},\ell+1)}$.
      \IF{$\bigl|\mathcal{PD}_{ij}^{(t^{\star},\ell+1)}\bigr| > 7\kb^{4\gamma}$}
        \STATE \textbf{break and move to the next $r$ value} \hfill \textit{\color{blue}// early stop due to PD-size overflow}
      \ELSE
        \STATE $\ell \gets \ell + 1$
      \ENDIF
    \ENDWHILE
    \STATE $\widehat{E}_{ij} \gets \mathcal{PD}_{ij}^{(t^{\star},\log_2 n_{ij})}$
    \STATE $\widehat{E} \gets  \widehat{E} \cup \{(\pi_{t^\star}^{(i,j)} \circ \lambda_{ij}^{-1}(x),\pi_{t^\star}^{(i,j)} \circ \lambda_{ij}^{-1}(y)): (x,y) \in  \widehat{E}_{ij}\}$ \hfill \textit{\color{blue}// recover the edges of $G_{ij}$}
  \ENDFOR
\ENDFOR
\STATE \textbf{Output:} $\widehat{G} = (V,\widehat{E})$.

\end{algorithmic}
\end{algorithm*}

\begin{lemma}
\label{lem:find_permutation_base}
Fix \((G,\Pi)\in \Tc(\epsilon_n,c,\gamma)\) and a pair \((i,j)\in[m]\times[m]\), meaning the subgraph  \(G_{ij}\) is also fixed.  For some $C_3 \ge 3e$, consider the $5C_3\kb_{ij} \log n$ (random) tests formed in Lines 9--12 in Algorithm \ref{alg:testing_permutation} for a fixed permutation index $t$ and repetition $r$.  Let $H$ denote the corresponding block graph at level $\ell_0$ (implicitly depending on $(i,j,t,r)$; see \eqref{eq:H_edge_set_t}).  Then, there exists a procedure for using the test results to decide whether \(\PV_{g_0}(H)=0\) or \(\PV_{g_0}(H) > 0\) that, with probability \(1 - \bigO \left(\frac{\kb^{2\gamma}}{n^{5}}\right)\), returns the correct answer and runs in time \(\bigO \big(\kb^{2\gamma}\log n\big)\).  The probability is taken over the randomness of the tests.
\end{lemma}

\begin{proof}

We proceed in several steps.

\textbf{Decision rule.}
At the base level \(\ell_0\), for each block \(U\) of \(H\), check the \(5\log n\) tests that contain \(U\).
Declare \(U\) non–defective if there exists at least one test containing \(U\) with a negative outcome.
Declare \(\PV_{g_0}(H)=0\) if every block is identified as non–defective; otherwise declare \(\PV_{g_0}(H)>0\).  Note that this procedure will always correctly declare $\PV_{g_0}(H)>0$ for defective blocks (since their tests results are all positive), so we only need to study the non-defective blocks.

\textbf{Analysis.}
Let \(U\) be a non–defective block at level \(\ell_0\).  
We will bound the probability that \(U\) is not identified. We proceed similarly to the proof of~\eqref{eq:bound-Y} in Lemma~\ref{lem:expectation}.
Consider a test that contains \(U\).
Let $\Bc_1$ be the event that some defective block among the $g_0$ blocks is included in the test, 
let $\Bc_2$ be the event that some defective pair $(U',V')$ with both $U'$ and $V'$ being (individually) non\mbox{-}defective is included in the test, 
and let $\Bc_3$ be the event that some block $V'$ among the $g_0$ blocks for which the pair $(U,V')$ is defective is included in the test. We have 
\begin{align}
    &\PP \left[ U \text{ is not identified in } 
       C_3\kb_{ij} \text{ tests} \right] \notag \\
    &\le\PP[\Bc_1] + \PP[\Bc_2] + \PP[\Bc_3]\, .
    \label{eq:U_block_not_identified}
\end{align}

Since \((G,\Pi)\in \Tc(\epsilon_n,c,\gamma)\), we have \(k_{ij}\le 2\kb_{ij}\), where \(k_{ij}\) denotes the number of edges in \(G_{ij}\).
Using the fact that the number of defective blocks among the \(g_0\) blocks is at most \(k_{ij}\), we get
\(\PP[\Bc_1] \le \tfrac{k_{ij}}{C_3\kb_{ij}} \le \tfrac{2}{C_3}\).
Moreover, there are at most \(k_{ij}\) edges in the block graph \(H\), and hence
\(\PP[\Bc_2] \le \tfrac{k_{ij}}{C_3^{2}\kb_{ij}^{2}} \le \tfrac{2}{C_3^{2}\kb_{ij}}\). Since the maximum degree of $H$ is at most $\sqrt{\kb_{ij}}$, we also have $\PP[\Bc_3] \leq \tfrac{\sqrt{\kb_{ij}}}{C_3\kb_{ij}}$. Therefore, by \eqref{eq:U_block_not_identified}, we have for sufficiently large $\kb_{ij}$ that
\[
\PP \left[ U \text{ is not identified in } C_3\kb_{ij} \text{ tests} \right] \le \frac{3}{C_3}\,.
\]
Hence, under the condition \(C_3 \ge 3e\), over \(5\log n\) independent rounds, the probability that a non–defective block \(U\) remains unidentified is at most
\(\Big(\tfrac{3}{C_3}\Big)^{5\log n} \le \tfrac{1}{n^{5}}\). 

\textbf{Running time.}
At level \(\ell_0\), there are \(g_0=\Theta(\kb^{2\gamma})\) blocks.
Each block is included in \(5\log n\) tests.
Thus, the total checking time is \(\bigO \big(\kb^{2\gamma}\log n\big)\).

\textbf{Overall success probability.}
By the bound above, the failure probability for a fixed non–defective block is \(\bigO \big(\tfrac{1}{n^{5}}\big)\).
A union bound over all \(g_0=\Theta(\kb^{2\gamma})\) blocks yields a total failure probability
\(\bigO \big(\tfrac{\kb^{2\gamma}}{n^{5}}\big)\), so the algorithm succeeds with probability at least
\(1-\bigO \big(\tfrac{\kb^{2\gamma}}{n^{5}}\big)\).

\end{proof}
\begin{lemma} \label{lem:pd_size_permutation_level}
Fix $(G,\Pi)\in \Tc(\epsilon_n,c,\gamma)$ along with a pair $(i,j)\in[m]\times[m]$ and $t\in[c]$ such that the block graph $H^{(t,\ell_0)}_{ij}$ (see \eqref{eq:H_edge_set_t}) satisfies
\begin{align*}
    &\PV_{g_0}\bigl(H^{(t,\ell_0)}_{ij}\bigr)=0, \\
    &\,\text{and}\,
       d_{g}\bigl(H^{(t,\ell)}_{ij}\bigr)\le \sqrt{\kb_{ij}} 
       \text{ for all } \ell \in \{\ell_0,\dots,\log_2 n_{ij} \},
\end{align*}
where $\ell_0\coloneqq \lceil \log_2 \kb^{2\gamma}\rceil$, $g_0 \coloneqq 2^{\ell_0}$, and $g = 2^\ell$. Consider the decoding procedure from Steps 5--14 of Algorithm \ref{alg:decoding_permutation} with some permutation $\pi_t^{(i,j)}$ on $G_{ij}$, from level $\ell_0$ up to level $\log_2 n_{ij}$. Let $\mathcal{PD}_{ij}^{(t,\ell)}$ denote the possibly defective set at level $\ell$ produced by these decoding steps. Then:
\begin{enumerate}
  \item Conditioned on the $\ell$-th level containing at most $7\kb^{4\gamma}$ possibly defective pairs, i.e., $|\mathcal{PD}_{ij}^{(t,\ell)}|\le 7\kb^{4\gamma}$, the number of possibly defective pairs at level $\ell+1$ is at most $7\kb^{4\gamma}$ (that is, $|\mathcal{PD}_{ij}^{(t,\ell+1)}|\le 7\kb^{4\gamma}$) with probability at least $1-\bigO \big( \tfrac{1}{\kb^{\gamma/2}}\big)$.
  \item With probability at least $1-\bigO \big( \tfrac{\log_2 n_{ij}}{\kb^{\gamma/2}} \big)$, the decoding time for triplet $(i,j,t)$ is $\bigO(\kb^{4.5\gamma}\log n_{ij})$.
\end{enumerate}
All probability statements are over the randomness of the test design.
\end{lemma}

\begin{proof}
For the first part, fix any $\ell \in \{\ell_0,\dots,\log_2 n_{ij}-1\}$ and assume that
$T \coloneqq \bigl|\mathcal{PD}_{ij}^{(t,\ell)}\bigr| \le 7\kb^{4\gamma}$. Consider the block graph $H_{ij}^{(t,\ell)}$. Let $\rE_{u,v}$ be the indicator random variable that a non-defective pair $(u,v) \in \mathcal{PD}_{ij}^{(t,\ell)}$ is not identified using the tests at level $\ell$.  Recall that we use the test design in Algorithm~\ref{alg:testing} on $G_{ij}$ with permutation $\pi_t^{(i,j)}$ starting from level $\ell_0$, and the expected number of edges satisfies $\kb_{ij} = q \binom{n_{ij}}{2} = \Theta(\kb^{\gamma})$. By identical reasoning to that of Lemma~\ref{lem:expectation} and Lemma~\ref{lem:variance}, we have
\begin{equation}
  \label{eq:permutation_expectation_variance}
  \begin{split}
    \EE \left[\sum_{u,v} \rE_{u,v}\right] &\le \frac{T}{24} \le \frac{\kb^{4\gamma}}{2}\,,\\
    \Var \left[\sum_{u,v} \rE_{u,v}\right] &\le \bigO \left(\frac{T^2}{\sqrt{\kb_{ij}}}\right) \;=\; \bigO \big(\kb^{7.5\gamma}\big)\,,
  \end{split}
\end{equation}
where the sum is over all non-defective pairs in $\mathcal{PD}_{ij}^{(t,\ell)}$.

We proceed using similar steps to the proof of Lemma~\ref{lem: bounds_size_PD}. Since $(G,\Pi)\in \Tc(\epsilon_n,c,\gamma)$, by Lemma~\ref{lem:typical_permutation} we have $k_{ij} \le 12\kb^{\gamma}$, where $k_{ij}$ denotes the number of edges of $G_{ij}$. Because $\PV_{g_0} \bigl(H^{(t,\ell_0)}_{ij}\bigr)=0$, it follows that $\PV_{g} \bigl(H^{(t,\ell)}_{ij}\bigr)=0$ for all $\ell \ge \ell_0$ (with $g = 2^{\ell}$). Hence, among the possibly defective pairs at level $\ell$, at most $k_{ij}$ are truly defective, which generate at most $6k_{ij} \le 72\kb^{\gamma}$ possibly defective pairs in $\mathcal{PD}_{ij}^{(t,\ell+1)}$. In addition, from~\eqref{eq:permutation_expectation_variance} and by Chebyshev’s inequality, with probability at least $1-\bigO \big( \tfrac{1}{\kb^{\gamma/2}}\big)$, at most $\kb^{4\gamma}$ non-defective pairs are incorrectly retained as possibly defective.  Since each such pair creates 6 children pairs, these contribute at most another $6\kb^{4\gamma}$ possibly defective pairs to $\mathcal{PD}_{ij}^{(t,\ell+1)}$, leading to a total of at most $6\kb^{4\gamma}+72\kb^{\gamma} \le 7\kb^{4\gamma}$ possibly defective pairs at level $\ell+1$ for sufficiently large $\kb$.

For the second part, the argument follows the proof of the decoding time in Theorem~\ref{thm:main}. At level $\ell_0$, we have $|\mathcal{PD}_{ij}^{(t,\ell_0)}| \le g_0^2 \le 4\kb^{4\gamma}$. From part~\textup{(1)}, by induction, for any level $\ell \ge \ell_0$, we have $|\mathcal{PD}_{ij}^{(t,\ell)}| \le 7\kb^{4\gamma}$ with conditional probability at least $1-\bigO \big( \tfrac{1}{\kb^{\gamma/2}}\big)$. Taking a union bound over at most $\log_2 n_{ij}$ levels, the same bound holds simultaneously for all levels with probability at least $1-\bigO \big( \tfrac{\log_2 n_{ij}}{\kb^{\gamma/2}} \big)$. Recall that the decoding time is dominated by the number of outcome checks in the decoding algorithm. At each level $\ell$, for each possibly defective pair in $\mathcal{PD}_{ij}^{(t,\ell)}$ we perform at most $\bigO(\sqrt{\kb_{ij}})=\bigO(\kb^{\gamma/2})$ outcome checks. Therefore, summing over at most $\log_2 n_{ij}$ levels, the total number of outcome checks is $\bigO(\kb^{4\gamma}\cdot \kb^{\gamma/2}\cdot \log n_{ij})=\bigO(\kb^{4.5\gamma}\log n_{ij})$, with probability at least $1-\bigO \big( \tfrac{\log n_{ij}}{\kb^{\gamma/2}}\big)$.
\end{proof}

Lemma~\ref{lem:pd_size_permutation_level} gives a bound on the failure probability for a fixed subproblem in a single repetition. Consequently, after \(c'\) independent repetitions, the failure probability for that fixed subproblem is $\bigO\!\left((\frac{\log n_{ij}}{\kb^{\gamma/2}})^{c'}\right)$. Since there are only \(\bigO(\kb^{1-\gamma})\) subproblems, taking \(c' > 2/\gamma\) ensures that a union bound over all subproblems yields an overall failure probability of \(o(1)\).


Based on the above lemmas, we obtain the following theorem.

\begin{theorem}\label{thm:permutation_method}
Fix \(\theta \in (0,1)\), and let \(G \sim \ER(n,q)\) for some \(q = \Theta\bigl(n^{-2(1-\theta)}\bigr)\). For every fixed constant
\(\gamma \in \bigl(0,\min\{1,\tfrac{1-\theta}{3\theta}\}\bigr)\), there exist constants \(c=c(\gamma)\), \(c'=c'(\gamma)\), \(C'=C'(\gamma,\theta)\), and absolute constants \(C_1,C_2,C_3\), all independent of \(n\) and \(\kb\), such that, when Algorithms~\ref{alg:testing_permutation} and~\ref{alg:decoding_permutation} are run with these parameters (where \(C_1\), \(C_2\), and \(C'\) are the constants
used in the subroutine Algorithm~\ref{alg:testing})\footnote{For the reader’s convenience: \(\gamma\) is the partition
parameter; \(c\) is the number of permutations; \(c'\) is the number of independent design repetitions used for amplification; \(C_1,C_2,C'\) are the constants inherited from the subroutine Algorithm~\ref{alg:testing}; and \(C_3\) is the base-level testing constant used to identify a favorable permutation at level \(\ell_0\). See Algorithm~\ref{alg:testing_permutation}, Definition~\ref{def:new_typical}, and
Lemmas~\ref{lem:typical_permutation},~\ref{lem:find_permutation_base}, and~\ref{lem:pd_size_permutation_level}, and recall also the notation summary in Table \ref{tab:sec4-notation}.} %
the algorithms use \(\bigO(\kb \log n)\) tests and, with probability \(1-o(1)\), guarantee the following:
\begin{enumerate}
  \item Algorithm~\ref{alg:decoding_permutation} returns \(\widehat{E} = E\).
  \item The decoding time is \(\bigO \big( \kb^{1+3.5\gamma} \log n \big)\). 
\end{enumerate}
\end{theorem}

\begin{proof}
By Lemma~\ref{lem:typical_permutation}, if \(c>1/\gamma\), then with probability \(1-o(1)\) as \(n\to\infty\) the pair \((G,\Pi)\)—where \(\Pi\) consists of \(c\) i.i.d.\ permutations drawn from a pairwise independent permutation family on \([n_{ij}]\)—belongs to \(\mathcal{T}(\epsilon_n,c,\gamma)\) (see Definition \ref{def:new_typical}, which also implies $G \in \mathcal{T}(\epsilon_n)$ from Definition \ref{def:typical}).
All subsequent analysis proceeds under the condition that \((G,\Pi) \in \Tc(\epsilon_n,c,\gamma)\).

\paragraph{Decoding time.}
Since \((G,\Pi) \in \Tc(\epsilon_n,c,\gamma)\), for each \((i,j) \in [m]\times[m]\) there exists \(t^{\star}\in[c]\) such that \(\PV_{g_0} \big(H^{(t^{\star},\ell_0)}_{ij}\big)=0\) and \(d_{g} \big(H^{(t^{\star},\ell)}_{ij}\big)\le \sqrt{\kb_{ij}}\) for all $\ell \in \{\ell_0,\dots,\log_2 n_{ij} \}$.
By Lemma~\ref{lem:find_permutation_base}, with probability at least \(1 - \bigO \left(\frac{\kb^{2\gamma}}{n^{5}}\right)\), we can find \(t^{\star}\) in time \(\bigO(\kb^{2\gamma} \log n)\).  Note that this probability scaling remains unchanged even after a union bound over $c$ permutations, because $c$ is constant.

Fix such \((i,j)\) and \(t^{\star}\).
For each test design for \(G_{ij}\) using \(\pi_{t^{\star}}^{(i,j)}\), Lemma~\ref{lem:pd_size_permutation_level} shows that the decoder succeeds with probability at least \(1 - \bigO \big( \tfrac{\log n_{ij}}{\kb^{\gamma/2}}\big)\) and runs in time \(\bigO(\kb^{4.5\gamma}\log n_{ij})\).
Algorithm~\ref{alg:decoding_permutation} then performs \(c'\) decoding executions on the \(c'\) given test designs for \(G_{ij}\) using \(\pi_{t^{\star}}^{(i,j)}\).
Therefore, with probability at least \(1-\bigO \bigl((\tfrac{\log n_{ij}}{\kb^{\gamma/2}})^{c'}\bigr)\), the decoding procedure on the pair \((i,j)\) with permutation \(\pi_{t^{\star}}^{(i,j)}\) succeeds with decoding time \(\bigO(\kb^{4.5\gamma}\log n_{ij})\). Recall also that we enforce an early-stopping rule: if, in any round \(r \in [c']\), the size of the possible defective set ever exceeds \(7\kb^{4\gamma}\), we immediately discard that round.

Finally, in Step~14 of Algorithm~\ref{alg:decoding_permutation}, we apply $\pi_{t^{\star}}^{(i,j)} \circ \lambda_{ij}^{-1}$ to the indices in the possible defective set to recover the edges of $G_{ij}$. Because the possible defective set at the final level has size at most \(7\kb^{4\gamma}\) (i.e., \(|\mathcal{PD}_{ij}^{(t^{\star},\log_2 n_{ij})}| \le 7\kb^{4\gamma}\)), this step runs in \(\bigO(\kb^{4\gamma})\) time: applying \(\pi_{t^{\star}}^{(i,j)} \circ \lambda_{ij}^{-1}\) to each element takes \(\bigO(1)\) time, since we can choose the pairwise independent permutation family constructed in Appendix~\ref{appendix:pairwise_independent}, which admits \(\bigO(1)\)-time evaluation and inversion on the word-RAM model.  Recall also that evaluating the relabeling map $\lambda_{ij}$ (or its inverse) from \eqref{eq:fix_labeling} takes $\bigO(1)$ time.

In summary, with probability at least
\(1-\bigO \bigl((\tfrac{\log n_{ij}}{\kb^{\gamma/2}})^{c'}\bigr)
 - \bigO \bigl(\tfrac{ \cdot \kb^{2\gamma}}{n^{5}}\bigr)\), the decoding procedure on the pair \((i,j)\) succeeds, yielding a total time per $(i,j)$ pair of
\[
\bigO\big(\kb^{4.5\gamma}\log n_{ij} \,+\, \kb^{4\gamma}\big) + \bigO \big(\kb^{2\gamma} \log n \big) \, .
\]

Aggregating over all \(\binom{m}{2}\leq \kb^{1-\gamma}\) block pairs gives and recalling that $n_{ij} = 2n\kb^{\frac{\gamma-1}{2}}$, the total decoding time is at most
\begin{align*}
    &\bigO\big(\kb^{1-\gamma}\cdot \kb^{4.5\gamma}\log n_{ij}\big) 
     + \bigO\big(\kb^{1-\gamma}\cdot \kb^{2\gamma}\log n\big) \\
    &\;=\; \bigO\big(\kb^{\,1+3.5\gamma}\log n\big),
\end{align*}
with probability at least \(1-o(1)\) when \(c>1/\gamma\) and \(c' > 2/\gamma\).

\paragraph{Error probability.}
For each $(i,j) \in [m]\times[m]$, let $\mathcal{PD}_{ij}^{(\log_2 n_{ij})}$ denote the possibly defective set obtained by the successful decoder on $G_{ij}$. 
Similar to~\eqref{eq:error_prob} in Theorem~\ref{thm:main}, we have
\begin{align*}
    &\PP \Big[\, E_{ij} = \widehat{E}_{ij}\ \bigm|\ 
       (G,\Pi) \in \Tc(\epsilon_n,c,\gamma), \big|\mathcal{PD}_{ij}^{(\log_2 n_{ij})}\big| 
       \le 7\kb^{4\gamma} \Big] \\
    &\;\ge\; 1 - \bigO \left( \frac{\kb^{4\gamma}}{n_{ij}^{C'}} \right),
\end{align*}
where $C' > 0$ is a parameter to Algorithm \ref{alg:testing}. 
We thus have the following high-probability behavior for each $(i,j) \in [m]\times[m]$, with similar probability terms to those in the decoding time analysis above: 
\begin{align*}
    &\PP \Big[\, \big|\mathcal{PD}_{ij}^{(\log_2 n_{ij})}\big| 
       \le 7\kb^{4\gamma}\ \bigm|\ 
       (G,\Pi) \in \Tc(\epsilon_n,c,\gamma) \Big] \\
    &\;\ge\; 1 
       - \bigO \left(\left(\frac{\log n_{ij}}{\kb^{\gamma/2}}\right)^{c'}\right) 
       - \bigO \left( \frac{\kb^{2\gamma}}{n^5} \right),
\end{align*}
for some $c' > 2/\gamma$. 
Combining these two findings and aggregating over all $\binom{m}{2} \le \kb^{1-\gamma}$ block pairs gives
\begin{align*}
    &\PP \big[\widehat{E} = E \,\bigm|\,
       (G,\Pi) \in \Tc(\epsilon_n,c,\gamma)\big] \\
    &\;\ge\; 1 - \bigO\left( \kb^{1-\gamma} \cdot \left( 
       \left(\frac{\log n_{ij}}{\kb^{\gamma/2}}\right)^{c'} 
       + \frac{\kb^{2\gamma}}{n^5} 
       + \frac{\kb^{4\gamma}}{n_{ij}^{C'}} \right) \right).
\end{align*}
By choosing $c' > 2/\gamma$ and $C' > \tfrac{8}{1 - \theta+ \gamma\theta}$, and recalling that $n_{ij} = 2n\,\kb^{(\gamma-1)/2}$ and $\kb = \Theta(n^{2\theta})$, we obtain
\[
\PP \big[\widehat{E} = E \,\bigm|\,(G,\Pi) \in \Tc(\epsilon_n,c,\gamma)\big] = 1 - o(1).
\] 
Since the event $(G,\Pi) \in \Tc(\epsilon_n,c,\gamma)$ holds with probability $1-o(1)$ (Lemma \ref{lem:typical_permutation}), it follows that 
\[
\PP \big[\widehat{E} = E\big] = 1 - o(1).
\]


\paragraph{Number of tests.} 
In Algorithm~\ref{alg:testing_permutation}, for each \((i,j)\in[m]\times[m]\), we apply the non-adaptive design of Algorithm~\ref{alg:testing} to \(G_{ij}\) under permutation \(\pi_t^{(i,j)}\), repeating it \(c'\) times.  At the base level \(\ell_0\), we first use \(5C_3\,\kb_{ij}\log n = \bigO(\kb^{\gamma} \log n) \) tests to determine whether the block graph \(H_{ij}^{(t,\ell_0)}\) contains no defective blocks. Then, at every level $\ell_0 \le \ell \le \log_2 n_{ij}-1$, the procedure uses $C_1C_2\kb_{ij} = \bigO(\kb^{\gamma})$ tests; at the final level $\log_2 n_{ij}$ it uses an additional $C_1C_2C'\kb_{ij}\log n_{ij} = \bigO(\kb^{\gamma}\log n_{ij})$ tests, where $C' > 8/(1-\theta +\gamma \theta)$ (with $\theta$ and $\gamma$ are fixed) is chosen to amplify the success probability. 
In addition, we have $c > 1/\gamma$ permutations, and for each $G_{ij}$ with permutation $\pi_t^{(i,j)}$ we use Algorithm~\ref{alg:testing} to design $c' > 2/\gamma$ independent test designs. 
Therefore, each pair $(i,j)$ leads to $\bigO \big(c\,c'\,\kb^{\gamma}\log n_{ij}\big) + \bigO(c \, c'\,\kb^{\gamma} \log n) \;=\; \bigO(\kb^{\gamma}\log n_{ij}+ \kb^{\gamma} \log n)$ tests, where we used the fact that $c,c'=\Theta(1/\gamma)$ and $\gamma$ is fixed. Summing over $\binom{m}{2} \le \kb^{1-\gamma}$ pairs and recalling $n_{ij} = 2n\,\kb^{(\gamma-1)/2}$, the total number of tests in Algorithm~\ref{alg:testing_permutation} is at most
\[
 \kb^{1-\gamma} \cdot \bigO \big(\kb^{\gamma}\log n_{ij} + \kb^{\gamma} \log n\big)
\;=\; \bigO \big(\kb \log n\big).
\]

\end{proof}

\section{Conclusion}

We have extended the fast binary splitting approach in group testing for non-adaptive learning of Erd\H{o}s--R\'enyi graphs, achieving a number of tests $\bigO(\kb \log n)$ and a decoding time of $\bigO(\kb^{1+\delta}\log n)$ for any fixed small $\delta>0$. This attains (nearly) the best of both objectives in terms of the number of tests and the decoding time. This suggests several possible directions for future research, including: (i) improving the decoding time, ideally reducing it to $\bigO(\kb \log n)$; (ii) attaining (significantly) improved constants in the number of tests (e.g., see \cite{wang2023quickly} in the context of group testing); (iii) extending this approach to the problem of learning random hypergraphs, as studied earlier in~\cite{austhof2025non}; and (iv) considering noisy versions of the problem of learning random graphs and hypergraphs.


\appendix
 \section{Concentration Inequalities}
\label{appendix:concentration}
In this section, we present the Chernoff bounds that are used throughout the paper. 
Let $X_1, X_2, \dots, X_n$ be a sequence of independent $\mathrm{Bernoulli}(q)$ random variables. 
Denote $X = \sum_{i=1}^{n} X_i$ and $\mu \coloneqq qn$. 
We then have the following (e.g., see \cite[Ch.~2]{boucheron2003concentration}):
\begin{enumerate}
    \item For any $\delta > 0$, 
    \begin{equation}
      \PP \left[X \ge (1+\delta)\mu\right] 
      \leq \exp \left(-\frac{\delta^2 \mu}{2+\delta} \right).
      \label{eq:chernoff1}
    \end{equation}
    \item For any $t \geq \mu$, 
    \begin{equation}
      \PP[X \ge t] 
      \leq \left( \frac{e\mu}{t} \right)^{t}.
      \label{eq:chernoff2}
    \end{equation}
\end{enumerate}

\section{Proof of Lemma \ref{lem:typical} (Typical Set Probability)} \label{sec:pf_typical}

The first condition in the definition of $\Tc(\epsilon_n)$ concerns the total number of edges $k = |E|$.  
Since $k$ follows a binomial distribution with mean $\kb = \binom{n}{2}q$, a standard concentration inequalities argument yields that \cite{li2019learning}
\[
  (1-\epsilon_n)\kb \leq k \leq (1+\epsilon_n)\kb
\]
with probability $1-o(1)$ as $n \to \infty$. Hence, condition (i) holds. For later use, we also note that, by applying the Chernoff bound (see~\eqref{eq:chernoff1} in Appendix~\ref{appendix:concentration}), we obtain
\begin{equation}
\label{eq:bound_number_edges}
    \PP[k > 2\kb] \leq \exp \left( - \frac{\kb^2}{6}\right) = \exp \left( - \Omega(n^{4\theta}) \right) \, 
\end{equation}
since $\kb = \Theta(n^{2\theta})$.

\medskip
We now turn to condition (ii).  
For each level $\ell$, let $g=2^{\ell}$ and consider the quantities $\PV_g$, $|E_g|$, and $d(G_g)$.  
These are random variables determined by $G$, and we prove that they satisfy the bounds in~\eqref{eq:number_edges_level}--\eqref{eq:degree_level} with probability at least $1-o(1)$ as $n \to \infty$.

{\bf Analysis of the final level.} For $\ell=\log_2 n$, we have $g=2^{\ell}=n$, and each $\Gc_{j}=\{j\}$ is a singleton. In this case, the block graph $G_g$ coincides  with graph $G$. We observe that~\eqref{eq:number_edges_level} and \eqref{eq:number_defective_level} hold with probability $1-o(1)$, with the former using \eqref{eq:bound_number_edges} and the latter being trivial since $\PV_g=0$.

For the maximum degree, as shown in~\cite{li2019learning}, with probability $1-o(1)$ we have
\[
  d(G_g)\le \dmax(G)\coloneqq
  \begin{cases}
     2nq, & \theta>\tfrac{1}{2},\\
     \log n, & \theta\le \tfrac{1}{2}.
  \end{cases}
\]
This is no higher than $\dmax$ defined in \eqref{eq:degree_level}, since $\frac{10 \kb}{g} \ge \frac{10 \kb}{n} = \frac{10 q}{n} {n \choose 2} \ge 5 nq$ (and $\sqrt{\kb} = \omega(\log n)$).  Therefore,~\eqref{eq:degree_level} holds with probability $1-o(1)$, which completes the analysis of the final level.  In the rest of the analysis, we consider~\eqref{eq:number_edges_level}--\eqref{eq:degree_level} for the other levels, i.e., for $g = 2^\ell$ with $ \ell \in \{\ceil{\log_2 \sqrt{\kb}}, \dots ,\log_2 n-1 \}$.

\smallskip
\textbf{Bounding $\PV_g$ and $|E_g|$.}  
For each $i \in [g]$, let $Z_i$ be the indicator that $\Gc_i$ is defective.  
By definition, $\PV_g = \sum_{i=1}^g Z_i$.  
Since $|\Gc_i| = n/g$, the number of possible edges inside $\Gc_i$ is $\binom{n/g}{2}$. We have

\[
\PP[Z_i = 1] =  1 - (1-q)^{\binom{n/g}{2}} \, .
\]
We observe (via $\frac{n}{2} \ge g \ge \sqrt{\kb}$ and $\kb = q{n \choose 2}$) that $q \binom{n/g}{2} = \kb/\binom{n}{2} \cdot \binom{n/g}{2} = \frac{\kb(n-g)}{g^2(n-1)}$. Therefore,  $ \frac{\kb}{2g^2} \le q \binom{n/g}{2} \le \frac{\kb}{g^2} \le 1$, which helps us bound $\PP[Z_i = 1]$ via the identity \( \frac{qT}{2} \leq 1 - (1-q)^{T} \leq qT  \) for any $0\leq qT \leq 1$. Specifically, combining this with $q \binom{n/g}{2} = \frac{\kb(n-g)}{g^2(n-1)}$ and $\frac{1}{2} \le \frac{n-g}{n-1} \le 1$, we obtain
\begin{equation}
     \label{eq:bound_prob_defective}
     \frac{\kb}{4g^2} \leq \PP[Z_i=1] \leq \frac{\kb}{g^2},
\end{equation}
and hence
\begin{equation}
  \label{eq:expected_PV}
     \frac{\kb}{4g} \leq \EE[\PV_g] \leq \frac{\kb}{g} \, .
\end{equation}
From the construction of $G_g$ and the fact that $G$ contains at most $k$ edges, we also have
\begin{equation}
  \label{eq:number_defective_pairs}
  |E_g| \leq g \cdot \PV_g + k. 
\end{equation}
We consider two cases:
\begin{itemize}
    \item {
\textbf{Case 1: $\theta > \tfrac{1}{2}$.}
Since $\kb=\Theta(n^{2\theta})$ and $g\in[\sqrt{\kb},\,n/2]$, there exist constants $c_1,c_2>0$ such that, for all sufficiently large $n$,
\[
  \frac{\kb}{g}\in\big[c_1\,n^{2\theta-1},\,c_2\,n^{\theta}\big].
\]
By the Chernoff bound (~\eqref{eq:chernoff1} in Appendix~\ref{appendix:concentration}), we obtain
\begin{align}
    \PP \left[\PV_g>\tfrac{2\kb}{g}\right]
    &\le \exp \Big(- \Omega \Big(\frac{\kb}{g} \Big)\Big) \notag \\
    &\le \exp \left(-\Omega(n^{2\theta-1}) \right)\, .
    \label{eq:PV_g_theta_1}
\end{align}

Consequently, from \eqref{eq:number_defective_pairs}, and $\eqref{eq:bound_number_edges}$, we have
\[
  \PP[|E_g| > 4\kb] \leq \exp \left( -\Omega(n^{2\theta-1}) \right) + \exp \left( - \Omega(n^{4\theta}) \right) \, .
\]
    }
    \item{
    \textbf{Case 2: $\theta \leq \tfrac{1}{2}$.}  
From \eqref{eq:number_defective_pairs} and the bound on $\EE[\PV_g]$ in \eqref{eq:expected_PV}, we obtain
\[
  \EE[|E_g|] \leq \EE[\PV_g]\cdot g + \EE[k] \leq 2\kb.
\]
Hence, by Markov’s inequality,
\[
  \PP[|E_g| > 2\kb\log^2 \kb] \leq \frac{1}{\log^2 \kb} = \bigO \left(\tfrac{1}{\log^2 n}\right).
\]
For $\PV_g$, at the final level with $g=2^{\ceil{\log_2 \sqrt{\kb}}}$, it is immediate that $\PV_g \leq g < 2\sqrt{\kb}$.  
For the other levels with $g \in [2\sqrt{\kb}, n/2]$, the Chernoff bound (~\eqref{eq:chernoff2} in Appendix~\ref{appendix:concentration}) along with~\eqref{eq:expected_PV} gives
\begin{align}
\PP \left[\PV_g > 2\sqrt{\kb}\right]
  &= \PP \left[\PV_g > 2\cdot \frac{g}{\sqrt{\kb}} \cdot \frac{\kb}{g}\right] \nonumber \\
  &\leq \PP \left[\PV_g > 2\cdot \frac{g}{\sqrt{\kb}} \cdot \EE[\PV_g]\right] \nonumber \\
  &\leq \left(\frac{e\sqrt{\kb}}{2g}\right)^{2\sqrt{\kb}}  \nonumber \\
  &\leq \left(\frac{e}{4}\right)^{2\sqrt{\kb}} \nonumber \\
  &= \exp \left(-\Omega(n^{\theta})\right). \label{eq:PV_g_theta_2}
\end{align}
    }
\end{itemize}

We notice that all of the above bounds decay to zero (significantly) faster than $\frac{1}{\log n}$, meaning that they are still $o(1)$ after a union bound over at most $\log_2 n$ values of $\ell$.  Thus, with probability $1-o(1)$,~\eqref{eq:number_edges_level} and \eqref{eq:number_defective_level} hold for all $\ell$.


\medskip
\textbf{Bounding $d(G_g)$.} For the first level with $g =  2^{\ceil{\log_2 \sqrt{\kb}}}$, it is immediate that $d(G_g) \leq g$, which we see is at most $\min \{2\sqrt{\kb}, \frac{10\kb}{g} \}$ by substituting this $g$ value.  It remains to consider the remaining levels after the first one.

Fix a level $\ell \in \{ \ceil{\log_2 \sqrt{\kb}}+1,\dots, \log_2 n -1 \}$, and consider a given block $\Gc_u$.  If this block is defective (i.e., has internal edges), then by definition it does not affect $d(G_g)$; thus, for our purposes here, we may condition on the event that $\Gc_u$ is non-defective.   
For each $i \in [g]$, let $Z_i$ be the indicator that $(\Gc_u,\Gc_i)$ is defective.  
Then $d_g(u) = \sum_{i=1}^g Z_i$.  
Having conditioned (implicitly) on $\Gc_u$ being non-defective, the number of possible edges in $\Gc_u \cup \Gc_i$ is $\tfrac{n^2}{g^2} + \binom{n/g}{2}$. We thus have 
\[
\PP[Z_i = 1] = 1 - (1-q)^{\tfrac{n^2}{g^2} + \binom{n/g}{2}} \, .
\] 
By similar reasoning to~\eqref{eq:bound_prob_defective}, for $n$ sufficiently large and $g \in [2\sqrt{\kb},n/2]$, we have
\[
  \frac{5\kb}{4g^2} \leq \PP[Z_i=1] \leq \frac{4\kb}{g^2}.
\]
Therefore, $\frac{5\kb}{4g} \leq \EE[d_g(u)] \leq \frac{4\kb}{g}$. We now consider two cases:
\begin{itemize}
    \item {
    \textbf{Case 1: $\theta > \tfrac{1}{2}$.} By the Chernoff bound, similar to~\eqref{eq:PV_g_theta_1}, we have
\[
  \PP \left[d_g(u) > \tfrac{10\kb}{g}\right] \leq \exp \left(-\Omega(n^{2\theta-1}) \right).
\]
Applying a union bound over at most $g \leq n$ nodes, we obtain
\[
  \PP[d(G_g) > \tfrac{10\kb}{g}] \leq n \cdot \exp \left(-\Omega(n^{2\theta-1}) \right).
\]
    }
\item{
\textbf{Case 2: $\theta \leq \tfrac{1}{2}$.}  
Analogous to~\eqref{eq:PV_g_theta_2}, we have
\[
  \PP[d_g(u) > 8\sqrt{\kb}] \leq \exp \left( - \Omega(n^\theta)\right).
\]
and taking a union bound over $g \leq n$ nodes gives
\[
  \PP[d(G_g) > 8\sqrt{\kb}] \leq n \cdot \exp \left( - \Omega(n^\theta)\right) \, .
\]
}
    
\end{itemize}


Since the above bounds decay to zero strictly faster than $\frac{1}{\log n}$, we may take a further union bound over all levels $\ell$ (with at most $\log_2 n$ levels).  It follows that with probability $1-o(1)$,~\eqref{eq:degree_level} holds for all $\ell$.

\medskip
Combining the above results, we conclude that all conditions in the definition of $\Tc(\epsilon_n)$ hold with probability at least $1-o(1)$ as $n \to \infty$.  This completes the proof.

\section{Construction of a Pairwise Independent Permutation Family}
\label{appendix:pairwise_independent}

In this appendix, we provide our required pairwise independent hash function that can be evaluated in $\bigO(1)$ time.  Throughout our paper, we consider a number of vertices equaling a power of two (rounding up via ``dummy vertices'' with no edges if needed), but it is instructive to first consider permuting a \emph{prime number} $p$ of elements, say labeled $\{0,1,\dotsc,p-1\}$.  In this case, a straightforward choice of random permutation is
\begin{equation*}
    \pi(x) = ax + b \quad (\mathrm{mod}~p),
\end{equation*}
where $(a,b)$ are uniformly random over the $p(p-1)$ combinations with $a \ne 0$.  Under this choice, it is straightforward to verify that for any $x\ne x'$, the pair $(\pi(x),\pi(x'))$ is uniform over the $p(p-1)$ possible pairs.  Specifically, this is seen by subtracting two equations $ax+b=y$ and $ax'+b=y'$ and using the invertibility of $x-x'$ (mod $p$) to determine $a$ from $(y,y')$, which in turn determines $b$, thus establishing that all $p(p-1)$ pairs $(y,y')$ are equally likely.

Powers of two can be handled in essentially the same manner, but using more general finite field arithmetic, namely, $\mathbb{F}_{2^{m}}$ instead of $\mathbb{F}_p$. While these ideas are not new, we provide a self-contained treatment for completeness.

\paragraph{Setting.}
Let \(N=2^{m}\) for some integer \(m\ge 1\). Fix the finite field \(\mathbb{F}_{2^{m}}\) via a monic irreducible polynomial \(f(t)\in \mathbb{F}_2[t]\) of degree \(m\), so that
\(\mathbb{F}_{2^{m}}\cong \mathbb{F}_2[t]/(f(t))\) with the polynomial basis \(\{1,t,\dots,t^{m-1}\}\).
Let \(\varphi:[N]\to \mathbb{F}_{2^{m}}\) be the bijection that maps the binary expansion
\[
x=\sum_{i=0}^{m-1} c_i 2^{i},\qquad c_i\in\{0,1\},
\]
to the corresponding polynomial-basis representation
\[
\varphi(x)\;=\;\sum_{i=0}^{m-1} c_i t^{i}\pmod{f(t)}.
\]
We write \(\varphi^{-1}\) for the inverse map: if
\(z=\sum_{i=0}^{m-1} c_i t^{i}\) with \(c_i\in\{0,1\}\), then
\(\varphi^{-1}(z)=\sum_{i=0}^{m-1} c_i 2^{i}\).
For background on the isomorphism \(\mathbb{F}_{2^m}\cong \mathbb{F}_2[t]/(f(t))\), see~\cite{mceliece2012finite}.

\paragraph{Construction.}
Let \(\mathbb{F}_{2^{m}}^{\times}\coloneqq \mathbb{F}_{2^{m}}\setminus\{0\}\).
For each \(a\in\mathbb{F}_{2^{m}}^{\times}\) and \(b\in\mathbb{F}_{2^{m}}\), define
\[
\widetilde{\pi}_{a,b}:\mathbb{F}_{2^{m}}\to \mathbb{F}_{2^{m}},\qquad
\widetilde{\pi}_{a,b}(z)=az+b,
\]
with operations taken in \(\mathbb{F}_{2^{m}}\).
Transporting to \([N]\) via \(\varphi\), define
\[
\pi_{a,b}\;:=\;\varphi^{-1}\circ \widetilde{\pi}_{a,b}\circ \varphi \;:\; [N]\to [N].
\]
Now set
\[
\mathcal{F}_N \;:=\; \{\pi_{a,b}\,:\, a\in \mathbb{F}_{2^{m}}^{\times},\; b\in \mathbb{F}_{2^{m}}\}.
\]
Since \(a\neq 0\), each \(\widetilde{\pi}_{a,b}\) is a bijection on \(\mathbb{F}_{2^{m}}\); hence \(\pi_{a,b}\) is a bijection on \([N]\).
Thus \(\mathcal{F}_N\subseteq \mathcal{S}_N\) and \(|\mathcal{F}_N|=N(N-1)\).

\begin{lemma}\label{lem:affine_pairwise_binary}
If \(\pi\) is sampled uniformly from \(\mathcal{F}_N\), then for all distinct \(x_1,x_2\in [N]\) and all distinct \(y_1,y_2\in [N]\),
\[
\PP_{\pi\sim \mathrm{Unif}(\mathcal{F}_N)} \bigl[\pi(x_1)=y_1 \, , \, \pi(x_2)=y_2\bigr]
\;=\; \frac{1}{N(N-1)}.
\]
As a consequence, \(\mathcal{F}_N\) is a pairwise independent family of permutations on \([N]\).
\end{lemma}
\begin{proof}
The proof is standard and can be found in most texts on the construction of pairwise independent hash families (e.g. see~\cite{luby2006pairwise}). We include it here for the reader’s convenience. Let \(z_i=\varphi(x_i)\) and \(w_i=\varphi(y_i)\) for \(i=1,2\). Then
\[
\pi(x_1)=y_1,\qquad \pi(x_2)=y_2
\]
are equivalent to
\[
\widetilde{\pi}_{a,b}(z_1)=w_1,\qquad \widetilde{\pi}_{a,b}(z_2)=w_2,
\]
i.e.,
\begin{equation}\label{eq:lin_system_binary}
a z_1 + b = w_1,\qquad a z_2 + b = w_2 \quad \text{in } \mathbb{F}_{2^{m}}.
\end{equation}
Because \(x_1\ne x_2\) we have \(z_1\ne z_2\), and because \(y_1\ne y_2\) we have \(w_1\ne w_2\).
Subtracting the two equations in \eqref{eq:lin_system_binary} gives
\[
a(z_1 - z_2) = w_1 - w_2,
\]
which has the unique solution
\[
a = \frac{w_1 - w_2}{\,z_1 - z_2\,} \in \mathbb{F}_{2^{m}}^{\times}.
\]
Then \(b\) is uniquely determined by \(b = w_1 - a z_1\). Thus there is exactly one pair
\((a,b)\in \mathbb{F}_{2^{m}}^{\times}\times \mathbb{F}_{2^{m}}\) that satisfies \eqref{eq:lin_system_binary}. Since \(|\mathcal{F}_N|=N(N-1)\) and the sampling is uniform, the probability equals \(1/(N(N-1))\).

\end{proof}

\paragraph{Evaluation and storage.}
We work in the word\mbox{-}RAM model: reading a single integer in \(\{0,\dots,N-1\}\) from memory and performing arithmetic on such integers each takes constant time.
Since elements of \(\mathbb{F}_{2^m}\) have bit\mbox{-}length \(m=\Theta(\log N)\), we can encode each element in \(\bigO(1)\) machine words under this model.
Each \(\pi_{a,b}\) is specified by \((a,b)\) (and we store \(a^{-1}\) once), so the description uses \(\bigO(\log N)\) bits.

\emph{Cost of \(\varphi\) and \(\varphi^{-1}\).}
Packing/unpacking the \(m\) base\mbox{-}2 digits of \(x\) (or the coefficient vector) and converting to/from the polynomial\mbox{-}basis representation modulo \(f(t)\) takes \(\bigO(1)\) time in this model.
Writing \(T(\cdot)\) for the computation time on a single input under the word\mbox{-}RAM model (unit\mbox{-}cost arithmetic on integers in \(\{0,\dots,N-1\}\)), it follows that
\[
T(\varphi)=T(\varphi^{-1})=\bigO(1).
\]

\emph{Field operations over \(\mathbb{F}_{2^m}\).}
With the above encoding, coefficient\mbox{-}wise XOR (addition) and one field multiplication followed by reduction modulo the fixed irreducible \(f(t)\) are each counted as \(\bigO(1)\) time on the word\mbox{-}RAM.
For any \(x\in[N]\),
\begin{align*}
    \pi_{a,b}(x) &= \varphi^{-1}\big(a\cdot \varphi(x)+b\big), \\
    \pi_{a,b}^{-1}(x) &= \varphi^{-1}\big(a^{-1}\cdot(\varphi(x)-b)\big).
\end{align*}
Each query performs one \(\varphi\) operation, one field multiplication, one field addition, and one \(\varphi^{-1}\) operation; hence,
\[
T\big(\pi_{a,b}(x)\big)=T\big(\pi_{a,b}^{-1}(x)\big)=\bigO(1).
\]

\emph{Storage.}
We can store \((a,b,a^{-1})\) using \(\bigO(\log N)\) bits.
The irreducible polynomial \(f(t)\in\mathbb{F}_2[t]\) of degree \(m\) requires \(m\) bits for its non\mbox{-}leading coefficients, i.e., \(\bigO(\log N)\) bits.  



%
%
%
%
%
%
%
%
%
%
%

\bibliographystyle{IEEEtran}
\bibliography{bibliofile}

\
\begin{IEEEbiographynophoto}{Hoang Ta}
received the Ph.D. degree in Computer Science from the École Normale Supérieure de Lyon, France, in 2022. He was a research fellow with the National University of Singapore, Singapore, from 2023 to 2025. He is currently a faculty member at Hanoi University of Science and Technology, Vietnam. His research interests include theoretical computer science and quantum information theory.
\end{IEEEbiographynophoto}

\begin{IEEEbiographynophoto}{Jonathan Scarlett}
 received the B.Eng.~degree in electrical engineering and the B.Sci.~degree in  computer science from the University of Melbourne, Australia. From October 2011 to August 2014, he was a Ph.D. student in the Signal Processing and Communications Group at the University of Cambridge, United Kingdom. From September 2014 to September 2017, he was post-doctoral researcher with the Laboratory for Information and Inference Systems at the \'Ecole Polytechnique F\'ed\'erale de Lausanne (EPFL), Switzerland. Since January 2018, he has been with the Department of Computer Science and Department of Mathematics at the National University of Singapore, where he is currently an Associate Professor. His research interests are in the areas of information theory, machine learning, signal processing, and high-dimensional statistics.
\end{IEEEbiographynophoto}

\end{document}